\documentclass[11pt,a4paper,onecolumn]{article}
\usepackage{times}
\usepackage{algorithmic}
\usepackage {graphicx}
\usepackage[symbol]{footmisc}
\usepackage{amsmath,amssymb,amsthm, epsfig}

\newcommand{\Draft}{0}   % 0 = final, anything else is rough draft

\ifthenelse{\Draft = 1} {
  \newcommand{\fixMe}[2][]{
    \typeout{***** ERROR: fixMe still in final version *****}
  }
} {
  \newcommand{\fixMe}[2][] {[{\bf #1}] {\bf \marginpar{\large FIX}} {\em #2}}

}

\newtheorem{Theorem}{Theorem}
\newtheorem{lemma}{Lemma}
\begin{document}

\title{Label-Guided Graph Exploration with Adjustable Ratio of Labels\thanks{This
work was done while Meng Zhang worked with Leszek G\c{a}sieniec at
University of Liverpool.  }}

\author{
Meng Zhang$^\mathtt{a}$\thanks{Corresponding author.\newline
\indent \hspace{0.5em}
\emph{E-mail addresses}: zhangmeng@jlu.edu.cn (M. Zhang), whdzy2000@vip.sina.com (Y. Zhang), jtang@cse.sc.edu (J. Tang) }, Yi Zhang$^\mathtt{b}$, Jijun Tang$^\mathtt{c}$\\
\small $^\mathtt{a}$\emph{College of Computer Science and
Technology, Jilin
University, Changchun, China}\vspace{-2pt} \\
\small $^\mathtt{b}$\emph{Department of Computer Science, Jilin
Business and
Technology College, Changchun, China}\vspace{-2pt}\\
\small $^\mathtt{c}$\emph{Department of Computer Science \&
Engineering, Univ. of South
Carolina, USA}\vspace{-2pt}\\
}

\date{}

\if 0
\IEEEauthorblockN{Meng Zhang\IEEEauthorrefmark{1}, Yi
Zhang\IEEEauthorrefmark{2} and Jijun Tang\IEEEauthorrefmark{3}}

\IEEEauthorblockA{\IEEEauthorrefmark{1}College of Computer Science
and Technology, Jilin
University, Changchun, China\\
Email: zhangmeng@jlu.edu.cn\\}
\IEEEauthorblockA{\IEEEauthorrefmark{2}Department of Computer
Science, Jilin Business and
Technology College, Changchun, China\\
Email: whdzy2000@vip.sina.com\\}
\IEEEauthorblockA{\IEEEauthorrefmark{3}Department of Computer
Science \& Engineering, Univ. of South
Carolina, USA\\
Email: jtang@cse.sc.edu} }\fi

\maketitle

\begin{abstract}{
The graph exploration problem is to visit all the nodes of a
connected graph by a mobile entity, e.g., a robot. The robot has
no a priori knowledge of the topology of the graph or of its size.
Cohen et~al.~\cite{Ilcinkas08} introduced label guided graph
exploration which allows the system designer to add short labels
to the graph nodes in a preprocessing stage; these labels can
guide the robot in the exploration of the graph. In this paper, we
address the problem of adjustable 1-bit label guided graph
exploration. We focus on the labeling schemes that not only enable
a robot to explore the graph but also allow the system designer to
adjust the ratio of the number of different labels. This
flexibility is necessary when maintaining different labels may
have different costs or when the ratio is pre-specified. We
present 1-bit labeling (two colors, namely black and white)
schemes for this problem along with a labeling algorithm for
generating the required labels. Given an $n$-node graph and a
rational number $\rho$, we can design a 1-bit labeling scheme such
that $n/b\geq \rho$ where $b$ is the number of nodes labeled
black. The robot uses $O(\rho\log\Delta )$ bits of memory for
exploring all graphs of maximum degree $\Delta$. The exploration
is completed in time
$O(n\Delta^{\frac{16\rho+7}{3}}/\rho+\Delta^{\frac{40\rho+10}{3}})$.
Moreover, our labeling scheme can work on graphs containing loops
and multiple edges, while that of Cohen et al. focuses on simple
graphs. }

\end{abstract}

\section{Introduction}

This paper concerns the task of graph exploration by a finite
automaton guided by a graph labeling scheme. A finite automaton
$\mathcal{R}$, called a robot, must be able to visit all the nodes
of any unknown anonymous undirected graph $G = (V,E)$. The robot
has no a priori information about the topology of $G$ and its
size. While visiting a node the robot can distinguish between the
edges that are incident on this node. At each node $v$ the edges
incident on it are ordered and labeled by consecutive integers $0,
\ldots, d-1$ called port numbers, where $d = \mathtt{deg}(v)$ is
the degree of $v$. We will refer to port ordering as a local
orientation. We use \emph{Mealy Automata} to model the robot. The
robot has a transition function $f$ and a finite number of states.
If the automaton in state $s$ knows the port $i$ through which it
enters a node of degree $d$, it switches to state $s'$ and exits
the node through port $i'$, that is, $f(s,i,d)=(s',i')$.

The graph exploration by mobile agents (robots) recently received
much attention, and different graph exploration scenarios have
been investigated. In the case of tree exploration, it is shown by
Diks et al.~\cite{DIKS04} that the exploration of $n$-node trees
such that the robot can stop once exploration is completed,
requires a robot with memory size $\Omega(\log\log\log n)$ bits,
and $\Omega(\log
 n)$ bits are necessary for exploration with return. Moreover, they constructed
an algorithm of exploration with return for all trees of size at
most $n$, using $O(\log^2 n)$ bits of memory. In the work of
Amb\"{u}hl et al.~\cite{Gas07}, the memory is lowered to $O(\log
n)$ bits for exploration with return. Flocchini et
al.~\cite{FIP10} later showed that a team of $\Omega(n)$
asynchronous oblivious robots are necessary for most $n$-node
trees, and that it is possible to explore the tree by $O(\log
n/\log\log n)$ robots only if the maximum degree of the tree is 3.

The memory size of the robot is widely adopted as the measurement
of the efficiency~\cite{Ko80,FI04,FIR05,FIP05,Re05}. Fraigniaud et
al.~\cite{FIP05} proved that a robot needs $\Theta(D\log \Delta)$
bits of memory to explore all graphs of diameter $D$ and maximum
degree $\Delta$. By the result of Reingold~\cite{Re05}, a robot
equipped with $O(\log n)$ bits of memory is able to explore all
$n$-node graphs in the perpetual exploration model, where the
return to the starting node is not required. The lower bound of
memory bits $\Omega(\log  n)$ is proved by Rollik~\cite{Ko80}.

In the scenario adopted in~\cite{ben02,FI04,FIR05}, the robot is
provided with a \emph{pebble} that can be dropped on a node and
used to identify the node later. The authors in~\cite{ben02}
showed that a robot can explore the graph with only one pebble if
it knows an upper bound on the number of nodes, otherwise
$\Theta(\log\log n)$ pebbles are necessary and sufficient.
Flocchini et~al. \cite{FMS09} studied a dynamic scenario where the
exploration is on a class of highly dynamic graphs. %, i.e. links
%between nodes exist only sometimes.

Recently, much research is focused on the exploration of anonymous
graphs guided by labeling the graph
nodes~\cite{KKP05,FIP08,Ilcinkas06,Ilcinkas08,GKM08,KM09,CDG09}.
The periodic graph exploration requires that the automaton has to
visit every node in an undirected graph infinitely many times in a
periodic manner. Ilcinkas~\cite{Ilcinkas06} considered minimizing
the length of the exploration period by appropriate assignment of
local port numbers. G\c{a}sieniec et al.~\cite{GKM08} improved the
upper bound of the exploration period $\pi$ from $4n - 2$ to
$3.75n - 2$ in an $n$-node graph, providing the agent with a
constant memory. For an oblivious agent, \cite{DJS05} achieved a
period of $10n$. Recently, Cyzyowicz et al.~\cite{CDG09} showed a
period of length at most $4 \frac 13 n$ for oblivious agents and a
period of length at most $3.5n$ for agents with constant memory.
Kosowski et al.~\cite{KM09} provided a new port labeling which
leads to shorter exploration cycles, improving the bound to
$\pi\leq 4 n - 2$ for oblivious agents.

Cohen et al.~\cite{Ilcinkas08} introduced the exploration labeling
schemes. The schemes consist of an algorithm $\mathcal{L}$ and a
robot $\mathcal{R}$ such that given any simple graph $G$ with any
port numbering, the algorithm $\mathcal{L}$ labels the nodes of
$G$, and $\mathcal{R}$ explores $G$ with the help of the labeling
produced by $\mathcal{L}$. It is shown that using only 2-bit
(actually, 3-valued) labels a robot with a constant memory is able
to explore all graphs, and the exploration is completed in time
$O(m)$ in any $m$-edge simple graph. The authors also presented a
1-bit labeling scheme (two kinds of labels, namely \textbf{black}
and \textbf{white}) on bounded degree graphs and an exploration
algorithm for the colored graph. The robot uses a memory of at
least $O(\log\Delta)$ bits to explore all simple graphs of maximum
degree $\Delta$. The robot stops once the exploration is
completed. The completion time of the exploration is
$O(\Delta^{O(1)}m)$.

\subsection{Our Results}
\if 0 In this paper, we consider a model where the system designer
is provided with a limited number of flags and is allowed to add
flags to graph nodes in the preprocessing stage; these flags can
guide the robot in the exploration of the graph. Maintaining of
flags may have costs (e.g., a lighting lamp), so it is necessary
to limit the number of flags.  Given an $n$-node graph $G$ and $b$
flags where $b \leq n/2$, we design a flagging scheme where the
number of flagged nodes is not greater than $b$. Following the
terms in~\cite{Ilcinkas08}, this scheme is a 1-bit labeling scheme
where each black node corresponds to a flagged node while each
white node corresponds to a node without a flag; the ratio of the
number of all the nodes to the number of black nodes is not less
than a given rational number $\rho=n/b$.\fi

We consider the problem of adjustable label guided graph
exploration. Since maintaining different labels may have different
costs, it is necessary to limit the number of some labels. For
example, in a 1-bit labeling scheme, if we use a lighting lamp to
represent `1' and a turned off lamp to represent `0', the number
of lighted lamps (label `1') may be limited to reduce the cost.
For a 1-bit labeling scheme on an $n$-node graph $G$ where the
number of nodes labeled black is $b$, we define
\emph{\textbf{$N$-ratio}} as the ratio of the number of nodes to
the number of nodes colored black, that is, $n/b$. Given a
rational number $\rho$, we can design a 1-bit labeling scheme on
$G$ such that the $N$-ratio is not less than $\rho$.

The 1-bit labeling scheme in~\cite{Ilcinkas08} does not guarantee
an arbitrary $N$-ratio and works specifically on simple graphs,
i.e., undirected graphs without loops or multiple edges. This
scheme employs the function of counting the number of neighbors
for a node, which is impossible in a non-simple graph with
multi-edges and loops. Using only the port numbering will not
allow a robot to know whether two neighbors of a node are the
same.

We present 1-bit labeling schemes that can adjust the $N$-ratio
and can work on non-simple graphs. We first investigate a family
of $N$-ratio tunable labeling schemes where the $N$-ratio can be
changed but not in a precise way. We classify the nodes in $G$ by
the distances between each node, and a specific node $r$ is
assigned as the \textbf{\emph{root}}. Each class of nodes in the
classification is called a \textbf{\emph{layer}}. In this family
of labeling schemes, all nodes in the same layer are labeled
similarly. We call $\rho'=bl/l$ the \emph{\textbf{$L$-ratio}} of
the labeling scheme where $l$ is the number of layers, and $bl$ is
the number of black layers. We introduce the $L$-ratio tunable
labeling schemes, enabling a robot to explore all graphs of
maximum degree $\Delta$. Starting from any node, the robot returns
to the root once the exploration is completed. We also design a
procedure for a robot to label the graph. But we need an extra
label to indicate that a node is not labeled yet.

Based on the $L$-ratio tunable labeling schemes, we introduce the
$N$-ratio adjustable labeling schemes. Precisely, given an
expected $N$-ratio $2\leq \rho\leq (D+1)/4$, we derive a series of
labelings from an $L$-ratio tunable labeling. Throughout the
paper, we use $\rho'$ to denote the $L$-ratio and $\rho$ to denote
the expected $N$-ratio. We prove that a labeling scheme with
$N$-ratio not less than $\rho$ can be found in these labeling
schemes. The exploration is completed in time
$O(n\Delta^{\frac{16\rho+7}{3}}/\rho+\Delta^{\frac{40\rho+10}{3}})$;
the robot need $O(\rho\log \Delta)$ bits of memory.

Table~\ref{table} compares our approach with the work of Cohen et
al.~\cite{Ilcinkas08}. In the case of $\rho=2$, our approach
extends the 1-bit labeling scheme in~\cite{Ilcinkas08} from simple
graphs to non-simple graphs. The exploration algorithms are
different, but their space and time complexities are similar for
simple graphs. When working on a simple graph labeled by the 1-bit
labeling scheme in~\cite{Ilcinkas08}, our exploration algorithm
runs in time $O(\Delta^{10}n)$ as in~\cite{Ilcinkas08}. Both
approaches derive a spanning tree from the graph by the labeling.
In~\cite{Ilcinkas08}, the tree contains all nodes; in our approach
the tree contains only black nodes, and the edges are paths of the
graph. To find a path of length $l$, the robot performs at most
$\Delta^{2l+2}$ traversals. Moreover, we use a new method to
identify the root and its neighbors for non-simple graphs.

%The upper bound on memory in our approach is close to the optimal.
When $\rho$ comes close to the diameter, the amount of memory used
by the robot is not far from that of the situation where all nodes
are white (that is, there is no labeling). It is known that
$\Omega(D \log \Delta)$~\cite{FIP05} bits of memory are necessary
without pre-labeling of the graph, which is the same bound as ours
when $\rho$ comes close to the diameter.

\begin{table*}[t]
\caption{Comparison of the labeling schemes in~\cite{Ilcinkas08}
and ours. The
first two rows are from~\cite{Ilcinkas08}. }%
\begin{center}
\small
\begin{tabular}{|c|c|c|c|c|}
\hline
Label size  & Robot's memory & Time & Ratio & Works on\\
(\#bits) & (\#bits) & (\#edge-traversals) & (\#black nodes) & \\
\hline $2$ & $O(1)$ & $O(m)$ & $-$ & simple graphs\\
\hline $1$ & $O(\log \Delta)$ & $O(\Delta^{O(1)}m)$&no
guarantee\footnotemark[3] & simple graphs\\
\hline 1&$O(\rho\log \Delta)$ & $O(n\Delta^{16\rho/3}/\rho+\Delta^{40\rho/3+1})$ & $ \leq n/\rho$ & non-simple graphs\\
\hline
\end{tabular}

\end{center}
\label{table}%
\end{table*}%
\footnotetext[3]{The number of black nodes in~\cite{Ilcinkas08}
can vary (without any control) from $\Theta(1)$ to $\Theta(n)$,
depending on the
  cases.}

\section{$L$-ratio Tunable 1-Bit Labeling Schemes for Bounded Degree Graphs  \label{S2}}

In this section, we describe an $L$-ratio tunable exploration
labeling scheme using 1-bit labels. Let $G$ be an $n$-node graph
 of degree bounded by $\Delta$. It is possible to color the
nodes of $G$ with two colors namely black and white, while the
$L$-ratio of the labeling is tunable. There exists a robot that
can explore the graph $G$ by the aid of the labeling, starting
from any node and terminating after identifying that the entire
graph has been traversed.

\subsection{Notions}

Let $v$ and $u$ be nodes connected by edge $e$. Denote by
$port(e,u)$ the port number of the port of $u$ which $e$ is
incident on. A path $P$ in a non-simple graph is defined as a
series of edges $e_0,e_1,\ldots,e_k$ such that for a series of
nodes $n_0,n_1, \ldots,n_{k+1}$, edge $e_i$ connects $n_i$ and
$n_{i+1}$ $(0\leq i\leq k)$. The string $p_0p_1\ldots p_{2k+1}$,
where $p_i=port(e_{\lfloor i/2\rfloor},n_{\lceil i/2\rceil})$
$(0\leq i\leq 2k+1)$, is called the \emph{label} of $P$. We denote
by $P^{-1}$ the reversal path of $P$. We say that a path $P$ is
greater than path $P'$, if the label of $P$ is lexicographically
greater than the label of $P'$. The distance between two nodes
$u$, $v$ is the number of edges in the shortest path from $u$ to
$v$, denoted by $d(u,v)$. Let $L_i$ denote the set of nodes that
are at distance $i$ from $r$, and $L_0=\{r\}$. For layers $L$ and
$L'$, we let $d(r,L)$ denote the distance between any node in
layer $L$ and $r$, $d(L,L')$ denote $|d(r,L)-d(r,L')|$.

\subsection{Labeling Schemes}

The following is a class of $L$-ratio tunable 1-bit labeling
schemes.

\noindent\emph{\textbf{Labeling $\mathcal{AL}$.}} Pick an
arbitrary node $r\in V$ and assign it the \textit{root} of
$\mathcal{AL}$. Label $r$ black. Select two different non-negative
integers $d_1$, $d_2$ satisfying $d_1\geq 2$ and $\lfloor
d_2/2\rfloor\geq d_1$. Define four classes of nodes $A,B,C$, and
$D$ as follows:
\begin{center}
\begin{flushleft}
$C=\{v\in V\mid d(r,v) \ \mathtt{ mod }\ (d_1+d_2+2)=0\}$,\\
$D=\{v\in V\mid d(r,v) \ \mathtt{ mod }\ (d_1+d_2+2)=1\}$, \\
$A=\{v\in V\mid d(r,v) \ \mathtt{ mod }\ (d_1+d_2+2)=d_2+1\}$,\\
$B=\{v\in V\mid d(r,v) \ \mathtt{ mod }\ (d_1+d_2+2)=d_1+d_2+1\}$.
\end{flushleft}
\end{center}
Label all the nodes in class $A,B,C$, and $D$ black and label all
the nodes left white. The $\mathcal{AL}$ labeling is denoted by
$\langle r,d_1,d_2 \rangle$.\vspace{4pt}

An example of $\mathcal{AL}$ labeling schemes is shown in
Figure~\ref{scheme}. A layer is called a white (black) layer if
all nodes in this layer are white (black). Denote black layers by
$\mathit{BL}_0, \mathit{BL}_1,\ldots, \mathit{BL}_{\mathit{D_B}}$,
where $\mathit{BL}_0=L_0$, $\mathit{D_B}+1$ is the number of black
layers, and $d(r,\mathit{BL}_i)<d(r,\mathit{BL}_j)$ if $i<j$. For
$X\in\{A,B,C,D\}$, layer $\mathit{BL}_i$ is said to be an
$X$-layer if $\mathit{BL}_i\subset X$. Two black layers are said
to be adjacent if one is $\mathit{BL}_i$ and another is
$\mathit{BL}_{i+1}$. The black nodes whose neighbors are all black
are called \textbf{B}-nodes.

\begin{figure}[htbp]
\centering
\begin{minipage}[t]{0.52\linewidth}
    \centering
    \includegraphics[width=2.9in]{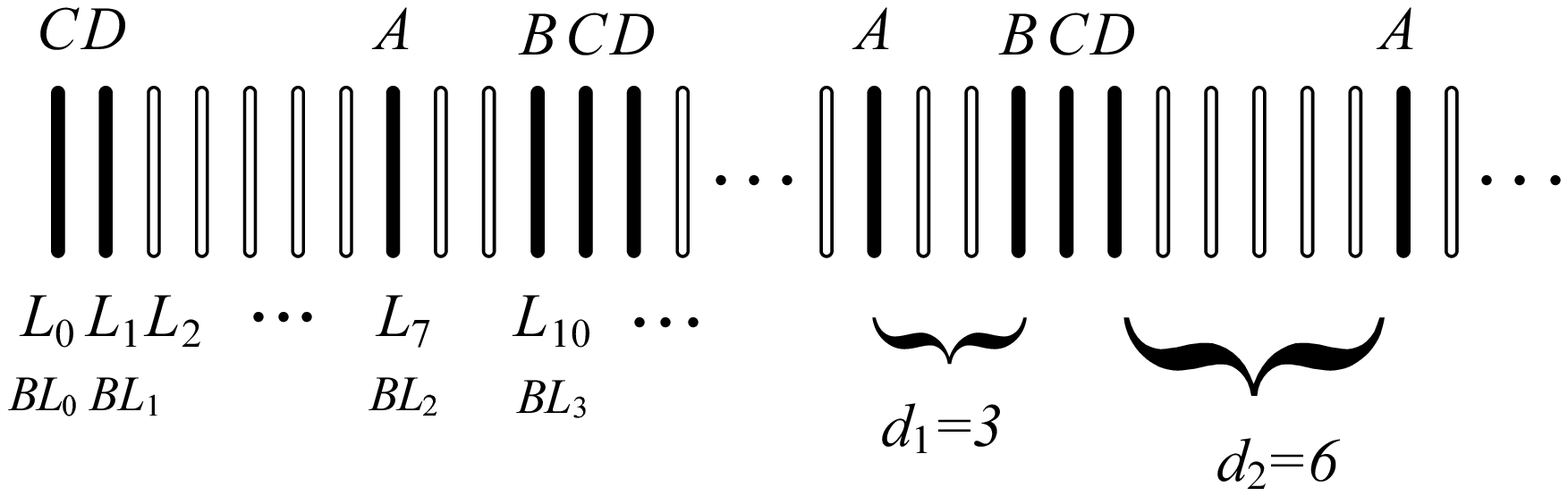}
    \caption{ An $\mathcal{AL}$ labeling scheme. Each line represents a layer. Black lines represent black layers, and white lines represent white layers.}%
    \label{scheme}
\end{minipage}
\hspace{13mm}
\begin{minipage}[t]{0.35\linewidth}
    \centering
    \includegraphics[width=1.5in]{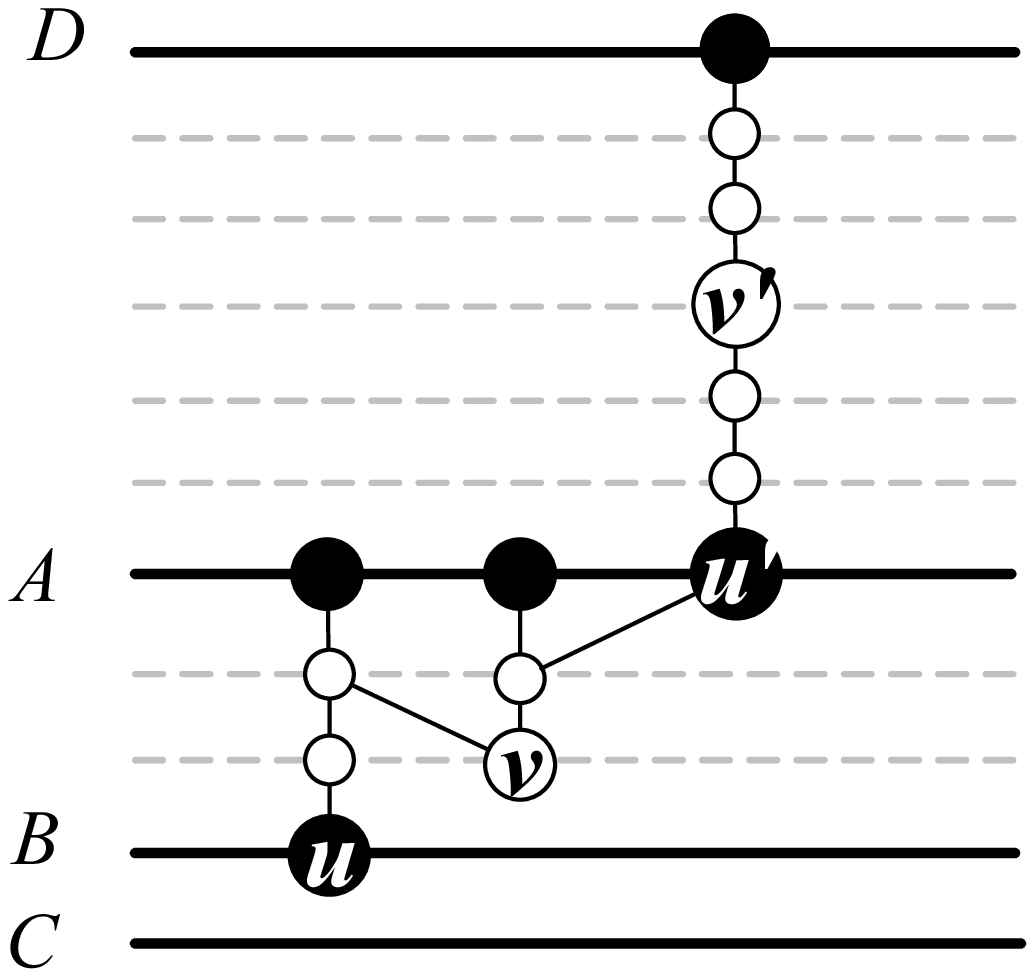}
    \caption{ $R_W(v)=1$, $R_W(v')=2$. By property 3, node $u$ and $u'$ can be distinguished by $R_W(v)$ and $R_W(v')$.}%
    \label{claim}
\end{minipage}

\end{figure}

The $L$-ratio of the labeling can be altered by adjusting $d_1$
and $d_2$, but it cannot be adjusted precisely to guarantee that
the $L$-ratio is not less than a given rational value. We assume
that $D\geq d_1+d_2+1$, that is, there are at least four black
layers. Then the upper bound on the $L$-ratio is $(D+1)/4$. The
minimal $L$-ratio is of an $\mathcal{AL}$ labeling where
$d_1=2,d_2=4$, and $D=9$, and there are six black layers in the
labeled graph. We have the $L$-ratio $\rho'\geq 5/3$.

For $\mathcal{AL}$ labeling schemes, we will prove the following
in the remaining of Section~\ref{S2}.

\begin{Theorem} Let $G$ be an $n$-node graph of degree bounded by
an integer $\Delta$, and let $G$ be labeled by an $\mathcal{AL}$
labeling scheme. There exists a robot that can explore the graph
$G$, starting from any given node and terminating at $r$ after
identifying that the entire graph has been traversed. The robot
has $O(\rho'\log\Delta)$ bits of memory, and the total number of
edge traversals by the robot is $O(\Delta^{12\rho'-9}n)+
o(\rho'\Delta n)$, where $\rho'$ is the $L$-ratio of the labeling.
\label{main}
\end{Theorem}

For a black node $u$, we identify two subsets of nodes that can be
reached by a path from $u$. For $u\in \mathit{BL}_{i}$ $(0< i\leq
\mathit{D_B})$, $pred(u)$ is the set of nodes in
$\mathit{BL}_{i-1}$ such that for any $x\in pred(u)$,
$d(u,x)=d(\mathit{BL}_i,\mathit{BL}_{i-1})$. For $u\in
\mathit{BL}_{i}$ $(0\leq i< \mathit{D_B})$, $succ(u)$ is the set
of nodes in $\mathit{BL}_{i+1}$ such that for any $x\in succ(u)$,
$d(u,x)=d(\mathit{BL}_{i+1},\mathit{BL}_{i})$. For root $r$, we
set $pred(r) = \varnothing$, and we have $succ(r) =
\mathit{BL_1}$. For $u\in \mathit{BL}_{\mathit{D_B}}$,
$succ(u)=\varnothing$.

In the following, we derive an \textit{implicit spanning tree} of
black nodes rooted at $r$ from an $\mathcal{AL}$ labeling scheme.
For $u\in \mathit{BL}_{i}$ $(0\leq i< \mathit{D_B})$, denote by
$succ\_path(u)$ the set of paths of length
$d(\mathit{BL}_i,\mathit{BL}_{i+1})$ whose starting node is $u$
and ending node is in $\mathit{BL}_{i+1}$. For $u\in
\mathit{BL}_{\mathit{\mathit{D_B}}}$, $succ\_path(u)=\varnothing$.
For $u\in \mathit{BL}_{i}$ $(0< i\leq \mathit{\mathit{D_B}})$,
denote by $\mathit{pred\_path}(u)$ the set of paths of length
$d(\mathit{BL}_i,\mathit{BL}_{i-1})$ whose starting node is $u$
and ending node is in $\mathit{BL}_{i-1}$. The path in
$\mathit{pred\_path}(u)$ with the lexicographically smallest label
is called the \emph{parent path} of $u$, denoted by
$\mathit{par\_path}(u)$. We set
$\mathit{pred\_path}(r)=\varnothing$. The ending node of
$\mathit{par\_path}(u)$ is called the \emph{parent} of $u$,
denoted by $\mathit{parent}(u)$. The set of nodes whose parent is
$u$ is denoted by $child(u)$. We have $child(u)\subseteq succ(u)$
and $\mathit{parent}(u)\in pred(u)$. The reversal paths of the
parent paths of the nodes in $child(u)$ are called \emph{child
paths} of $u$. All black nodes, their parent paths, and their
child paths form an implicit spanning tree.

\subsection{Properties of $\mathcal{AL}$ Labeling Schemes}

In this section we describe three properties on $\mathcal{AL}$
labeling schemes. These properties are the basis of the
exploration algorithm. Since for any node $u$ there is a shortest
path from $u$ to $r$, we have the following property.

\newtheorem{Property}{Property}

\begin{Property}
Let $u\neq r$ be a node, and let $L_i$ be a black layer such that
$i<d(r,u)$. There exists at least a node $x\in L_i$ such that
$d(x,u)=d(r,u)-i$. \label{back}
\end{Property}

A useful corollary of Property~\ref{back} is that any class $D$
node has a $\textbf{B}$-node neighbor.

Assume that the nearest black nodes to some node $v$ are at
distance $\ell$. Then the \emph{white-radius} of $v$ is $\ell-1$,
denoted by $R_W(v)$. Property~\ref{maxR} gives the upper bound on
the white radius of white nodes between two adjacent black layers.
Figure~\ref{claim} gives an example.

\begin{Property}
Let $u$ be a white node, and let $d(r,\mathit{BL}_i)<d(r,u)<
d(r,\mathit{BL}_{i+1})$. We have $R_W(u)\leq
d(\mathit{BL}_i,\mathit{BL}_{i+1})-2$. \label{maxR}
\end{Property}

Let $P$ be a path from $u$ to $v$ of length $\ell$ where only $u$
and $v$ are allowed to be black. Path $P$ is called a
\emph{white-path} from $u$, or precisely, an
\emph{$\ell$-white-path}. Let $u\in \mathit{BL}_i$ $(i\neq 0)$,
and let $\ell=d(\mathit{BL}_i,\mathit{BL}_{i-1})$. According to
Property~\ref{back}, there is at least one $\ell$-white-path from
$u$ to a node in $\mathit{BL}_{i-1}$. The maximal white radius of
nodes in this path is $\lfloor \ell/2\rfloor-1$, which leads to
the following property.

\begin{Property}
Let $u\in \mathit{BL}_i$ $(i\neq 0)$, and let
$\ell=d(\mathit{BL}_i,\mathit{BL}_{i-1})$. There exists a white
path from $u$ that reaches a white node whose white radius is not
less than $\lfloor \ell/2\rfloor-1$. \label{maxR+}
\end{Property}

These properties are used in our exploration algorithm. For
example, we can distinguish between a class $A$ node and a class
$B$ node by applying these properties. For $u\in A$, there exists
a white node $x$ that can be reached by a white path from $u$ such
that $R_W(x)=\lfloor d_2/2\rfloor-1$. But for a class $B$ node
$u$, the maximal white radius of white nodes that can be reached
by a white path from $u$ is not greater than $d_1-2$. Since
$d_1\leq \lfloor d_2/2\rfloor$ (see the definition of the
$\mathcal{AL}$ labeling), $d_1-2$ is less than $\lfloor
d_2/2\rfloor-1$. Figure~\ref{claim} gives an illustration.

\subsection{The Local Search Procedure\label{ls}}

The following local search procedure can be used to visit all
nodes at distance not greater than a given radius from a
node.\vspace{6pt}

\noindent \textbf{Procedure} $\mathit{LocalSearch}(u,\ell,inport)$

\noindent\textbf{Input}: $u$ is the starting node, $\ell$ is the
radius, and $inport$ is the port from which $\mathcal{R}$ enters
$u$.
\begin{algorithmic}[1]

\STATE \textbf{if} \ \ {$\ell=0$} \ \ \textbf{then}
$report(u)$\footnote[4]{When the robot reports a node, it does not
exit from the procedure nor makes any movement.}

\STATE \textbf{else}

\STATE \ \ \ \textbf{for}{ $outport$ from 0 to $\mathtt{deg}(u)-1$
and $outport\neq inport$ } \textbf{do}

\STATE  \  \ \ \ \ \ $v\leftarrow$ the neighbor of $u$ which
$outport$ leads to

\STATE   \ \ \ \ \ \  $\mathcal{R}$ moves to $v$

\STATE   \ \ \ \ \ \  $inport'\leftarrow$ the port from which
$\mathcal{R}$ enters $v$

\STATE  \ \ \ \ \ \  $LocalSearch(v,\ell-1,inport')$

\STATE  \ \ \ \ \ \  $\mathcal{R}$ moves back to $u$

\STATE \textbf{return}
\end{algorithmic}

 By the call
$\mathit{LocalSearch}(u,\ell,-1)$, the robot explores all
neighbors of $u$ up to distance $\ell$. In the local search from
$u$ within radius $\ell$, there are at most
$\mathit{LS}(\ell)=2\Delta\Sigma_{i=0}^{\ell-1}(\Delta-1)^i=O(\Delta^{\ell})$
edge traversals, and at most $\Delta(\Delta-1)^{\ell-1}$ nodes are
reported. Note that an edge may be visited more than once, and a
node could be reported more than once. The robot is in node $u$
when the procedure terminates. We summarize the results on the
$\mathit{LocalSearch}$ procedure in the following lemma.

\begin{lemma}
In the local search from node $u$ within radius $\ell$, a robot
with $O(\ell\log\Delta)$ bits of memory visits all nodes at
distance not greater than $\ell$ from $u$ without visiting any
other node. There are at most $O(\Delta^{\ell})$ edge traversals
and at most $\Delta(\Delta-1)^{\ell-1}$ nodes being reported. The
robot is in node $u$ when the local search terminates. \label{LSP}
  \end{lemma}

We can revise the procedure to explore only the paths that are
greater than a given path $P$ from $u$ as follows. The robot first
moves to the end of $P$ via $P$ and restores the context of the
procedure for $P$ in its memory and then starts the procedure.

\subsection{Exploration Guided by Labeling\label{exp}}

The overall exploration performed by the robot is a depth first
search (DFS) of the implicit spanning tree. All nodes will be
visited in the DFS. The robot maintains a state
$s\in\{\mathtt{up}, \mathtt{down}\}$. Initially, $\mathcal{R}$ is
at the root $r$ of an $\mathcal{AL}$ labeling and leaves $r$ by
the port numbered 0 in state $\mathtt{down}$. Assume that
$\mathcal{R}$ enters a black node $u$ via a path $P$ that belongs
to the implicit spanning tree. If $\mathcal{R}$ is in state
$\mathtt{down}$, it searches for the minimal child path of $u$. If
$\mathcal{R}$ is in state $\mathtt{up}$, it moves down to the
starting node of $P$ and searches for the minimal child path of
$u$ that is greater than $P^{-1}$. In both cases, if $\mathcal{R}$
does not find the desired path, $\mathcal{R}$ moves to
$\mathit{parent}(u)$ via the parent path of $u$ and transits the
state to $\mathtt{up}$; otherwise $\mathcal{R}$ moves to the end
node of the path found
 and transits the state
to $\mathtt{down}$. The correctness of these
  procedures will be proved later.

To know whether a path belongs to the spanning tree, we use the
following procedures.
\begin{enumerate}
\item $\mathit{Get\_Par\_Path}(u)$ identifies the parent path of
$u\notin\{r\}\cup \mathit{BL}_1$ and $\mathit{parent}(u)$. If
$v=\mathit{parent}(u)$ is found, the procedure returns $v$, and
$\mathcal{R}$ has moved to $v$ and recorded the parent path of $u$
in its memory; otherwise the procedure returns ``false". \item
$\mathit{Next\_Child\_Path}(u,P)$ identifies the minimal child
path from $u\neq r$ that is greater than $P$ where $P$ is a child
path of $u$ or $\varnothing$\footnotemark[5]. When such a child
path, say $P'$, is found, the procedure returns the end of $P'$,
and $\mathcal{R}$ has moved to the end of $P'$. If no path is
found, the robot goes back to $u$, and the procedure returns
``false". \footnotetext[5]{$P$ can be replaced by the label of $P$
as the initial node of $P$ is also input.}
\end{enumerate}

All these procedures use a revised local search procedure, namely
\emph{white local search}. Given a radius $d$, a node $u$, and a
path $P$ from $u$\footnotemark[5], the white local search
procedure enumerates all the $d$-white-paths from $u$ that are
greater than $P$. It returns ``true" if such path exists and
``false" otherwise. In both cases, the robot is in $u$ when the
procedure terminates. This procedure is derived from
$\mathit{LocalSearch}$, and the following line should be inserted
into $\mathit{LocalSearch}$ between line 2 and line 3.

\begin{algorithmic}[ ]

\STATE \textbf{if} {$u$ is black and $\ell\neq$ the initial radius
of the local search} \textbf{then}\ \textbf{return}

\end{algorithmic}

This procedure has the same property as Lemma~\ref{LSP}. The term
``local search" refers to the white local search procedure in the
remainder of the paper.

\subsubsection{Procedure $\mathit{Get\_Par\_Path}$ and $\mathit{Next\_Child\_Path}$\label{Proc}}

We first present procedures that will be used many times in the
exploration procedures.
%By Jijun: may revise the above paragraph, a bit too simple

\textbf{Procedure $\mathit{Is\_B}$}

The $\mathit{Is\_B}$ procedure takes as input a black node $x$
that belongs to class $B$, $C$, or $D$ and returns ``$B$" iff $x$
is in class $B$. The robot first checks whether $x$ is a
\textbf{B}-node. If it is, $\mathit{Is\_B}(x)$ returns
``\textbf{B}-node". If not, the robot performs a local search from
$x$ within radius $d_1$ (denote the local search by
$\mathit{LS}_1$). Once a black node $y$ that has no
\textbf{B}-node neighbor is reported, $\mathit{Is\_B}(x)$ returns
``$B$".
%(It includes a local search from $y$ within radius 1 ($\mathit{LS}_2$) and
%a local search within radius 1 to check whether each neighbor is a
%\textbf{B}-node ($\mathit{LS}_3$)).
If no such black node $y$ is reported or no node is reported,
$\mathit{Is\_B}(x)$ returns ``$D$". In any case, $\mathcal{R}$ is
in node $x$ when the procedure returns.

\textbf{Procedure $\mathit{C\_or\_D}$}

The $\mathit{C\_or\_D}$ procedure takes as input a black node $x$
($x\neq r$) that belongs to class $C$ or $D$ and returns the class
in which $x$ is. If $x$ is not a \textbf{B}-node
$\mathit{C\_or\_D}(x)$ returns ``$D$". Otherwise the robot
performs a local search from $x$ within radius $1$ (denoted by
$\mathit{LS}_1$). For each black neighbor $y$ of $x$ reported,
perform $\mathit{Is\_B}(y)$. Once $\mathit{Is\_B}(y)$ returns
``$B$", $\mathit{C\_or\_D}(x)$ returns $``C"$. If for every $y$,
$\mathit{Is\_B}(y)$ does not return ``$B$", then
$\mathit{C\_or\_D}(x)$ returns $``D"$. In any case, $\mathcal{R}$
is in node $x$ when the procedure returns.

\textbf{Procedure $\mathit{A\_or\_B}$}

The $\mathit{A\_or\_B}$ procedure takes as input a black node $x$.
If $x$ belongs to
class $A$ or $B$, $\mathit{A\_or\_B}(x)$ returns the class to which $x$ %in which belongs? to which belongs? By Jijun
belongs. The robot performs a local search (denoted by
$\mathit{LS}_1$) within radius $d_1$ from $x$.
$\mathit{A\_or\_B}(x)$ returns ``$A$" if in this local search a
white node is reported whose white radius is $d_1-1$ and returns
``$B$" otherwise. \if 0 If $x$ is in class $D$ and $x$ has
children in the spanning tree, $\mathit{A\_or\_B}(x)$ will return
``$A$".\fi
 In any case, $\mathcal{R}$ is in node $x$ when the
procedure returns. \vspace{6pt}

Now we present procedure $\mathit{Get\_Par\_Path}$ and
$\mathit{Next\_Child\_Path}$.

%\newpage
\textbf{Procedure $\mathit{Get\_Par\_Path}(u)$}

Assume that $\mathcal{R}$ starts from a node $u\in \mathit{BL}_i$
$(i\geq 2)$. $\mathcal{R}$ aims at identifying the parent path of
$u$ and moving to $\mathit{parent}(u)$. According to the class
%By Jijun: the above may need some revision
that $u$ belongs to, we consider four cases. In the following,
``$X\rightarrow Y$" means that $\mathcal{R}$ is in an $X$-layer
node $u$ and tries to move to $\mathit{parent}(u)$ in the adjacent
$Y$-layer. In each case, the robot first calls procedure
$\mathit{Path\_Enumeration}$ (PE for short) and then calls
procedure $\mathit{Node\_Checking}$ (NC for short) for each path
enumerated by PE. The functions of these two procedures are: (1)
PE: Enumerating (reporting) a set of white paths comprising
$\mathit{pred\_path}(u)$ and their ends. (2)~NC: Checking whether
a path enumerated by PE is in $\mathit{pred\_path}(u)$. Since the
local search enumerates paths in lexicographic order, in the
following cases, the node in $pred(u)$ firstly found by NC is
$\mathit{parent}(u)$, and the path recorded by the robot is the
parent path of $u$. Figure~\ref{B2A} gives an illustration.

\begin{figure*}[ht]%
   \hbox to\textwidth{\hfil\includegraphics[width=\textwidth]{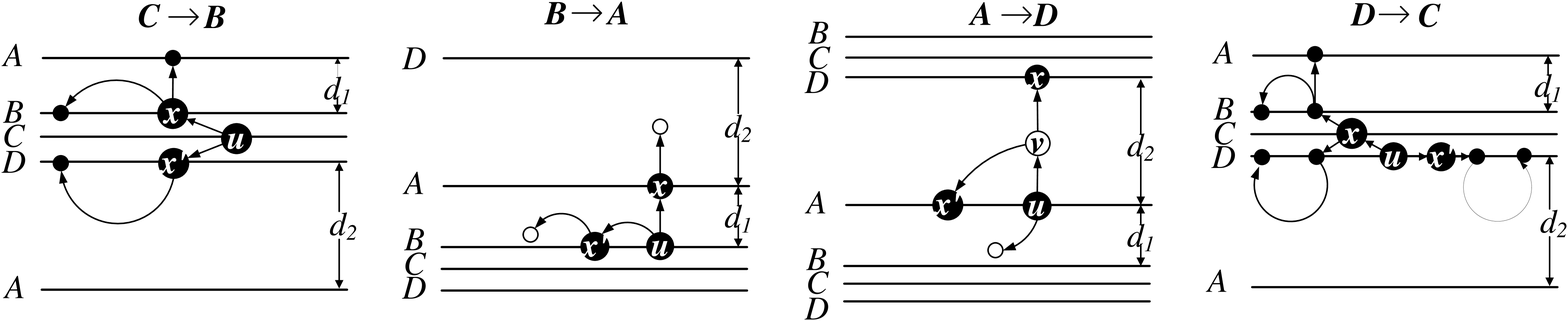}\hfil}%

   \caption{ Four cases in $\mathit{Get\_Par\_Path}$.
   The automaton starts from node $u$. It can reach node $x$ and $x'$ by PE. Node $x$ is in $pred$ set of node $u$ while $x'$ is not.  }%
   \label{B2A}%
%\end{figure*}%
%\begin{figure*}[ht]%
   \hbox to\textwidth{\hfil\includegraphics[width=\textwidth]{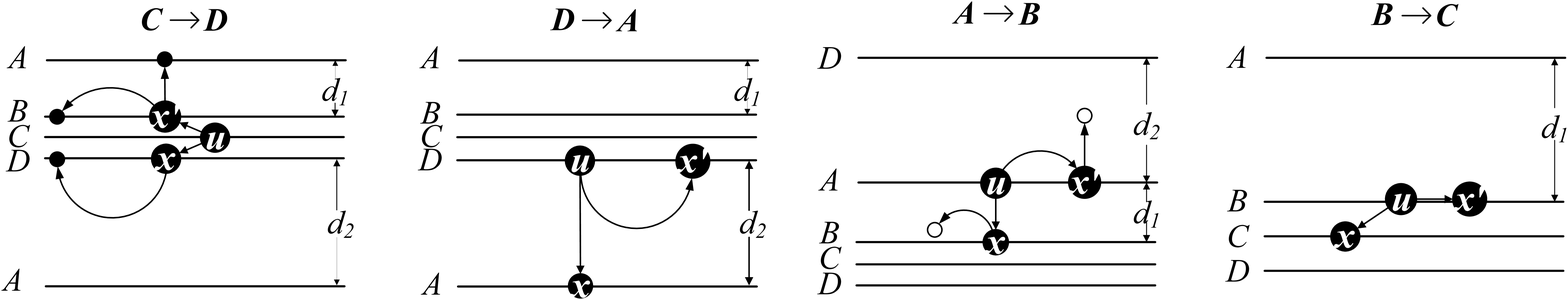}\hfil}%
   \caption{Four cases in $\mathit{Next\_Child\_Path}$.
   The automaton starts from node $u$. It can reach node $x$ and $x'$ by PE. Node $x$ is in $succ$ set of node $u$ while $x'$ is not.  }%
   \label{A2xB}%
\end{figure*}%

\noindent \textbf{ Case(1)} $C\rightarrow B$

\noindent PE: Perform a local search from $u$ within radius 1
($\mathit{LS}_1$).

\noindent NC: For each black node $x$ reported by PE, call
$\mathit{Is\_B}(x)$. Once $\mathit{Is\_B}(x)$ returns``$B$", we
return $x$.
% i.e., $x$ is $\mathit{parent}(u)$. Otherwise, we go back to PE and try another $x$.

\noindent \textbf{ Case(2)} $B\rightarrow A$

\noindent PE: Perform a local search from $u$ within radius $d_1$
($\mathit{LS}_1$).

\noindent NC: For each black node $x$ reported by PE, call
$\mathit{A\_or\_B}(x)$. Once $\mathit{A\_or\_B}(x)$ returns
``$A$", we return $x$.

\noindent \textbf{ Case(3)} $A\rightarrow D$

\noindent PE: (i) Perform a local search from $u$ within radius
$d_1$ ($\mathit{LS}_1$). (ii) From each white node $v$ reported,
perform a local search within radius $d_1-1$ ($\mathit{LS}_2$).
(iii) If all nodes visited in $\mathit{LS}_2$ are not black,
perform local search within radius $d_2-d_1$ from $v$
($\mathit{LS}_3$).

\noindent NC: For each black node $x$ reported by PE, if $x$ has a
\textbf{B}-node neighbor, we return $x$.

\noindent \textbf{ Case(4)} $D\rightarrow C$

\noindent PE: Perform a local search from $u$ within radius 1
($\mathit{LS}_1$).

\noindent NC: For each black node $x$ reported by PE, call
$\mathit{C\_or\_D}(x)$. Once $\mathit{C\_or\_D}(x)$ returns
``$C$", we return $x$.

In the above cases, if $x$ is returned by NC, then $x$ is
$\mathit{parent}(u)$, the path recorded in $\mathcal{R}$ is the
parent path of $u$, and the robot has moved to
$\mathit{parent}(u)$; otherwise we go back to PE to enumerate
another $x$ for NC.
 \vspace{4pt}

\noindent\textit{Identification of the Root and $\mathit{BL}_1$
Nodes.\label{IDr}} In $D\rightarrow C$, we distinguish a $C$-layer
node from a $D$-layer node by checking whether the node has a
neighbor in a $B$-layer. Since the nodes in
$\mathit{BL}_1\cup\{r\}$ have no ancestor in any $B$-layer, for
$u\in \mathit{BL}_1\cup\{r\}$, $D\rightarrow C$ in
$\mathit{Get\_Par\_Path}(u)$ will fail to find the parent of $u$.
For any other nodes, $\mathit{Get\_Par\_Path}$ will succeed in
finding their parents. Thus if $\mathit{Get\_Par\_Path}(u)$ fails,
then $u$ is in $\mathit{BL}_1\cup\{r\}$. The next problem is how
the robot identifies the root. The solution is that when leaving
the root, the robot memorizes the ports in the arrived nodes by
which it should return to the root. We revise
$\mathit{Get\_Par\_Path}(u)$ as follows. If $D\rightarrow C$ fails
to find $\mathit{parent}(u)$, we have $u\in
\mathit{BL}_1\cup\{r\}$; $\mathcal{R}$ goes to $r$ through the
port it memorized, and the procedure returns $r$.

 \vspace{6pt}

\textbf{Procedure $\mathit{Next\_Child\_Path}(u,P)$}

For a node $u\in \mathit{BL}_i$ $(i\geq 0)$, and a path $P$ from
$u$, the procedure identifies the minimal child path of $u$
greater than $P$. The robot calls the $\mathit{Enumerating}$
procedure to enumerate some paths from $u$ greater than $P$ and
calls the $\mathit{Identifying}$ procedure to check whether an
enumerated path is a child path. If such a path is found, we
return its end node; otherwise we return ``false", and the robot
backtracks to $u$, that is, $P$ is the maximal child path of $u$.

\if 0 For a node $u\in \mathit{BL}_i$ $(i\geq 0)$, and a path $P$
from $u$, the procedure identifies the minimal child path of $u$
that is greater than $P$. The robot first calls procedure
$\mathit{Enumerating}$, then calls procedure
$\mathit{Identifying}$ for each path enumerated by
$\mathit{Enumerating}$. The functions of the two procedures are:
(1) $\mathit{Enumerating}$: (i) PE: Using a local search to
enumerate all $d$-white-paths from $u$ that are greater than $P$,
where $d=d(\mathit{BL}_{i+1},\mathit{BL}_{i})$. (ii) NC: Checking
whether a node enumerated by PE is in $succ(u)$. (2)
$\mathit{Identifying}$: Check whether the path leading to a node
is a child path. If so, we return the node. If all children of $u$
are visited or $u$ has no child in the spanning tree, we return
``false", and the robot backtracks to $u$. \fi

\vspace{4pt} \noindent$\mathit{Enumerating}$. The procedure
contains two parts: (i) PE: Use a local search to enumerate all
$d$-white-paths $P'$ from $u$ that are greater than $P$, where
$d=d(\mathit{BL}_{i},\mathit{BL}_{i+1})$. If $P'$ does not exist,
return ``false". (ii) NC: Check whether the end node of $P'$ is in
$succ(u)$, if so, return this node. We consider the following
cases. Figure~\ref{A2xB} gives an illustration.

\noindent\textbf{ Case(1)}. $C\rightarrow D$

\noindent PE: Perform a local search from $u$ within radius 1
starting from $P$ ($\mathit{LS}_1$).

\noindent NC: For each black node $x$ reported by PE, call
$\mathit{Is\_B}(x)$. If ``$D$" is returned, we return $x$. If
``\textbf{B}-node" is returned, we call $\mathit{C\_or\_D}(x)$; if
``$D$" is returned, we return $x$.

\noindent\textbf{ Case(2)}. $D \rightarrow A$

\noindent PE: Perform a local search from $u$ within radius $d_2$
starting from $P$ ($\mathit{LS}_1$).

\noindent NC: For each black node $x$ reported by PE, if $x$ has
no \textbf{B}-node neighbor, we return $x$.

\noindent\textbf{ Case(3)}. $A\rightarrow B$

\noindent PE: Perform a local search from $u$ within radius $d_1$
 starting from
$P$ ($\mathit{LS}_1$).

\noindent NC: For each black node $x$ reported by PE, call
$\mathit{A\_or\_B}(x)$. Once ``$B$" is returned, we return $x$.

\noindent\textbf{ Case(4)}. $B\rightarrow C$

\noindent PE: Perform local search from $u$ within radius 1
starting from $P$ ($\mathit{LS}_1$).

\noindent NC: The black nodes without any white neighbor reported
by PE are in $succ(u)$
 ($\mathit{LS}_2$). We return the first such node.

In the above cases, if $x$ is returned by NC, then $x$ is in
$succ(u)$; otherwise we check another $x$ reported by PE.
 \vspace{4pt}

\noindent$\mathit{Identifying}$. When $\mathit{Enumerating}$ has
found a shortest path $P'$ from $u$ to a node $x$ in $succ(u)$,
$\mathcal{R}$ has moved to $x$ and recorded $P'$ in its memory.
$\mathcal{R}$ then checks whether this path is a child path of
$u$. If it is, the parent path of $x$ should be $P'^{-1}$. We use
\textbf{$\mathit{Check\_Par\_Path}(x,P'^{-1})$} to verify it. The
$\mathit{Check\_Par\_Path}$ procedure is similar to the
$\mathit{Get\_Par\_Path}$ procedure except that the former's PE
part is performed in decreasing lexicographic order. If
$\mathit{Check\_Par\_Path}(x,P'^{-1})$ finds a node in $pred(x)$,
then $P'^{-1}$ is not the parent path of $x$ and
$\mathit{Check\_Par\_Path}$ returns ``false"; otherwise $P'^{-1}$
is the parent path of $x$, and $\mathit{Check\_Par\_Path}$ returns
``true". In both cases, $\mathcal{R}$ is in node $x$ when
$\mathit{Check\_Par\_Path}$ terminates. If $P'$ is not a child
path of $u$, we go back to $\mathit{Enumerating}$ to enumerate
another path. If a child path of $u$ is found, then
$\mathit{Next\_Child\_Path}(u,P)$ returns the end node of $P$;
otherwise returns ``false". \vspace{6pt}

\textbf{Exploration from an Arbitrary Node}

When starting from an arbitrary node $x$, the robot should first
find the root. If $x$ is a white node, the robot performs a normal
local search within radius $d_2-1$ from $x$ and stops when
reaching a black node $u$ ($u$ is not a \textbf{B}-node). If $x$
is a \textbf{B}-node, the robot performs a normal local search
from $x$ within radius $2$ and stops when reaching a black node
that is not a \textbf{B}-node. A \textbf{B}-node is either in
class $C$ or in class $D$. For a \textbf{B}-node in class $C$, a
non-\textbf{B}-node black node will be reached by a local search
within radius 1; for a \textbf{B}-node in class $D$, such node
will be reached by a local search with radius 2. Therefore, in all
cases the robot can reach a black node $u$ that is not a
\textbf{B}-node. The robot then identifies in which class $u$ is.
If $u$ has a \textbf{B}-node neighbor, the robot performs
$\mathit{Is\_B}(u)$. If ``$B$" is returned, then $u\in B$. If
``$D$" is returned, then $u\in D$. (Note that $\mathit{Is\_B}(u)$
cannot answer ``\textbf{B}-node".) If $u$ does not have a
\textbf{B}-node neighbor, the robot calls $\mathit{A\_or\_B}(u)$.
We have $u\in A$ if ``$A$" is returned and $u\in B$ if ``$B$" is
returned.

After knowing the class of the starting node, the robot calls
procedure $\mathit{Get\_Par\_Path}$ all the way to find the root.
But our exploration cannot identify $r$ without memorizing the
port returning $r$. Fortunately, the robot knows whether it is in
a $\mathit{BL}_1$ node from the previous section. Let $r'$ be the
first $\mathit{BL}_1$ node found by the robot. Then $r$ is one of
the \textbf{B}-node neighbors of $r'$. We use every
\textbf{B}-node neighbor of $r'$ as a root and perform
explorations from them. The robot should memorize the port by
which it will return to $r'$. At least, the exploration rooted at
$r$ will be performed that visits all the nodes in $G$. The number
of edge traversals in this case is at most $\Delta$ times as large
as that in exploration from $r$.

\subsection{Correctness of the Exploration}

\begin{lemma}
For a black node $x$ that belongs to class $B$, $C$, or $D$,
$\mathit{Is\_B}(x)$ returns ``$B$" iff $x$ is in class $B$ and
returns ``$D$" iff $x$ is in $D$ and not a \textbf{B}-node. The
robot is in node $x$ when $\mathit{Is\_B}(x)$ exits. In the call
$\mathit{Is\_B}(x)$, the robot needs $O(d_1\log\Delta)$ bits of
memory, and the total number of edge traversals of the robot is
$O(\Delta^{d_1+2})$. \label{isb}
\end{lemma}
\begin{proof}
If $x\in B$, $\mathit{LS}_1$ will report at least one node in
class $A$ according to Property~\ref{back}. Since any node in
class $A$ has no \textbf{B}-node neighbor and has a white
neighbor, $\mathit{LS}_1$ will report at least one such node.
Therefore, if $x$ is a node in class $B$, $\mathit{Is\_B}(x)$
returns ``$B$".

If $x$ is a \textbf{B}-node, $\mathit{Is\_B}(x)$ returns
``\textbf{B}-node". It can be easily verified that $x$ is in class
$C$ or in class $D$. Let $x$ be a class $D$ node and not a
\textbf{B}-node. We have that either $\mathit{LS}_1$ does not
report any node or any node reported by $\mathit{LS}_1$ belongs to
class $D$. Since $d_2> d_1$, $\mathit{LS}_1$ will not reach any
class $A$ node. By Property~\ref{back}, any node in class $D$ has
at least a neighbor in class $C$ (that is a \textbf{B}-node), then
any black node reported by $\mathit{LS}_1$ has a \textbf{B}-node
neighbor, and thus $\mathit{Is\_B}(x)$ returns ``$D$". Therefore,
for a node $x$ that belongs to class $B$, $C$, or $D$,
$\mathit{Is\_B}(x)$ returns ``$B$" iff $x$ is in class $B$ and
``$D$" iff $x$ is in $D$ and not a \textbf{B}-node.

By definition, $\mathcal{R}$ is in node $x$ when
$\mathit{Is\_B}(x)$ exits. The number of edge traversals of
$\mathit{Is\_B}(x)$ is not greater than that of a local search
from $x$ within radius $d_1+2$. By Lemma~\ref{LSP}, there are at
most $O(\Delta^{d_1+2})$ edge traversals in the call
$\mathit{Is\_B}(x)$, and the memory space of $\mathcal{R}$ is
$O(d_1\log\Delta)$ bits.

\end{proof}

\begin{lemma}
For a black node $x\neq r$ that belongs to class $C$ or $D$,
$\mathit{C\_or\_D}(x)$ returns ``$C$" iff $x\in C$ and ``$D$" iff
$x\in D$. The robot is in node $x$ when $\mathit{C\_or\_D}$ exits.
 In the call
$\mathit{C\_or\_D}(x)$, the robot needs $O(d_1\log\Delta)$ bits of
memory, and the total number of edge traversals of the robot is
$O(\Delta^{d_1+3})$. \label{cord}
\end{lemma}
\begin{proof} If $x\in C\setminus\{r\}$, then $x$ has a neighbor $y$ in
class $B$, and thus $\mathit{Is\_B}(y)$ returns ``$B$" by
Lemma~\ref{isb}. Therefore, $\mathit{C\_or\_D}(x)$ returns ``$C$".
If $x\in D$, then all neighbors of $x$ belong to class $C$ or $D$.
Thus, for any neighbor $y$ of $x$, $\mathit{Is\_B}(y)$ does not
return ``$B$" by Lemma~\ref{isb}. Therefore,
$\mathit{C\_or\_D}(x)$ returns ``$D$".

By definition, $\mathcal{R}$ is in node $x$ when
$\mathit{C\_or\_D}(x)$ exits. The number of edge traversals of
$\mathit{C\_or\_D}(x)$ is not greater than that of a local search
from $x$ within radius $d_1+3$. By Lemma~\ref{LSP}, there are at
most $O(\Delta^{d_1+3})$ edge traversals in the call
$\mathit{C\_or\_D}(x)$, and the memory space of $\mathcal{R}$ is
$O(d_1\log\Delta)$ bits.
\end{proof}

\begin{lemma}
For a black node $x$ that belongs to class $A$ or $B$,
$\mathit{A\_or\_B}(x)$ returns ``$B$" if $x$ is in class $B$, and
returns ``$A$" otherwise. The robot is in $x$ when
$\mathit{A\_or\_B}$ exits. In the call $\mathit{A\_or\_B}(x)$, the
robot needs $O(d_1\log\Delta)$ bits of memory, and the total
number of edge traversals of the robot is $O(\Delta^{2d_1-1})$.
\if 0 If $x\in D$ and $child(x)\neq\varnothing$, then
$\mathit{A\_or\_B}(x)$ returns ``$A$".\fi \label{aorb}
\end{lemma}
\begin{proof}
According to Property~\ref{maxR+}, for $x\in A$, there exists a
white node whose white radius is not less than $\lfloor
d_2/2\rfloor-1$ that can be reached by a white path from $x$.
According to Property~\ref{maxR}, for $x\in B$, the white radius
of nodes that have a white path to $x$ are not greater than
$d_1-2$. By $\mathcal{AL}$, we have $\lfloor d_2/2\rfloor\geq
d_1$. Therefore, we know whether $x$ is in class $A$ or in class
$B$ by checking the maximal white radius of nodes that have a
white path to $x$. Thus, if $x\in A$, $\mathit{LS}_1$ will find a
node with white radius $d_1-1$, and $\mathit{A\_or\_B}(x)$ returns
``$A$"; if $x\in B$, no such node will be found, and
$\mathit{A\_or\_B}(x)$ returns ``$B$". \if 0 For $x\in D$ and
$child(x)\neq\varnothing$, there is a $d_2$-white-path from $x$ to
a child of $x$, then there is a white node in this path whose
white radius is $d_1-1$. $\mathit{A\_or\_B}(x)$ returns ``$A$".\fi

By definition, $\mathcal{R}$ is in node $x$ when
$\mathit{A\_or\_B}(x)$ exits. The number of edge traversals in the
call $\mathit{A\_or\_B}(x)$ is not greater than that of the local
search from $x$ within radius $2d_1-1$. By Lemma~\ref{LSP}, there
are at most $O(\Delta^{2d_1-1})$ edge traversals in the call
$\mathit{A\_or\_B}(x)$, and the memory space of $\mathcal{R}$ is
$O(d_1\log\Delta)$ bits.
\end{proof}

\begin{lemma}
For a black node $u$, let $\mathcal{R}$ know to which class $u$
belongs. For $u\notin \mathit{BL}_1\cup\{r\}$, when
$\mathit{Get\_Par\_Path}(u)$ exits, $\mathit{parent}(u)$ is
returned, and $\mathcal{R}$ is in $\mathit{parent}(u)$ and
recorded the parent path of $u$ in its memory. For $u\in
\mathit{BL}_1$, $\mathit{Get\_Par\_Path}(u)$ can identify that $u$
is in $\mathit{BL}_1$ and makes $\mathcal{R}$ return to $r$. In a
call to $\mathit{Get\_Par\_Path}$, there are at most
$O(\Delta^{d_2+2})$ edge traversals, and the robot needs
$O(d_2\log\Delta)$ bits of memory space. \label{GetPar}
\end{lemma}
\begin{proof} Let $\mathcal{R}$ initiate at a node $u\in \mathit{BL}_i$,
$i\geq 2$, when $\mathit{Get\_Par\_Path}(u)$ is called. We check
separately the four cases in the procedure. In each case, two
parts are to be proved: (1) PE can enumerate all paths in
$\mathit{pred\_path}(u)$ and their ends; (2) NC can identify
whether the end nodes reported by PE are in $pred(u)$.

\textbf{ Case(1)}. $u\in C$ ($C\rightarrow B$)

In this case, $pred(u)$ is a subset of the neighbors of $u$, since
$d(\mathit{BL}_{i},\mathit{BL}_{i-1})=1$. Therefore, all paths in
$\mathit{pred\_path}(u)$ can be enumerated by PE, so do the nodes
in $pred(u)$.

Any neighbor $x$ of $u$ belong to class $B$ ($\mathit{BL}_{i-1}$),
class $C$ ($\mathit{BL}_{i}$), or class $D$ ($\mathit{BL}_{i+1}$).
By Lemma~\ref{isb}, if and only if $\mathit{Is\_B}(x)$ returns
``$B$", $u$ is in class $B$, i.e., $pred(u)$.

\textbf{ Case(2)}. $u\in B$ ($B\rightarrow A$)

In this case, $d(\mathit{BL}_{i},\mathit{BL}_{i-1})=d_1$. By the
local search from $u$ within radius $d_1$ (PE), all paths in
$\mathit{pred\_path}(u)$ and all nodes in $pred(u)$ can be
reported.

The reported black nodes belong to class $A$ ($\mathit{BL}_{i-1}$)
or $B$ ($\mathit{BL}_{i}$); among them only the nodes in class $A$
are in $pred(u)$. According to Lemma~\ref{aorb}, if and only if
$\mathit{A\_or\_B}(x)$ returns ``$A$" then $u$ is in class $A$.
Thus by calling $\mathit{A\_or\_B}$ the nodes in $pred(u)$ can be
identified.

\textbf{ Case(3)}. $u\in A$ ($A\rightarrow D$)

In this case, $d(\mathit{BL}_{i},\mathit{BL}_{i-1})=d_2$. By
$\mathit{LS}_1$ and $\mathit{LS}_2$, all white nodes at distance
$d_1$ from both $u$ and $BL_i$ can be reported. As
$d(\mathit{BL}_{i},\mathit{BL}_{i+1})=d_1$, these white nodes are
between $\mathit{BL}_{i}$ and $\mathit{BL}_{i-1}$. The paths in
$\mathit{pred\_path}(u)$ containing such a white node can be
enumerated by $\mathit{LS}_3$. Since every path in
$\mathit{pred\_path}(u)$ contains such a white node, PE can
enumerate all paths in $\mathit{pred\_path}(u)$ and their ends.

\if 0 By $\mathit{LS}_1$, the white node $v$ in the parent path of
$u$ that is at distance $d_1$ from $u$ will be reported. As the
white radius of $v$ is $d_1-1$, in $\mathit{LS}_2$, no black node
will be visited. In $\mathit{LS}_3$, the parent path of $u$ and
$\mathit{parent}(u)$ will be reported. \fi

A black node $x$ reported by PE belongs to class $D$ or $A$.
Through the observations on $\mathcal{AL}$, any node in class $D$
has at least one \textbf{B}-node neighbor, and any node in class
$A$ has no \textbf{B}-node neighbor. So if $x$ has a
\textbf{B}-node neighbor then $x$ belongs to class $D$. Therefore
nodes in $pred(u)$ can be identified.

\textbf{ Case(4)}. $u\in D$ ($D\rightarrow C$)

In this case, $pred(u)$ is a subset of the neighbors of $u$. PE
can enumerate all paths in $\mathit{pred\_path}(u)$ and all nodes
in $pred(u)$.

For any neighbor $x$ of $u$, $x\in C$ or $x\in D$. By
Lemma~\ref{cord}, $\mathit{C\_or\_D}(x)$ can determine whether $x$
is in class $C$ which means $x\in pred(u)$.

\vspace{6pt} All the local searches in this procedure are
performed in increasing lexicographic order. According to
$\mathcal{AL}$, in the above cases, the node in $pred(u)$ first be
found is $\mathit{parent}(u)$, and the path stored in the memory
of $\mathcal{R}$ is the parent path of $u$. Since
$\mathit{parent}(u)$ exists, $\mathit{Get\_Par\_Path}(u)$ returns
$\mathit{parent}(u)$.

For $u\in \mathit{BL}_1\cup \{r\}$, let $\mathcal{R}$ take $u$ as
a class $D$ node. We can verify that $\mathit{Is\_B}(u)$ returns
``$D$", and $D\rightarrow C$ will fail to find the parent path of
$u$. By the above discussion, for nodes $x\notin \mathit{BL}_1\cup
\{r\}$, $\mathit{Get\_Par\_Path}(x)$ returns $\mathit{parent}(x)$.
Thus $\mathit{Get\_Par\_Path}(u)$ identifies that $u$ is in
$\mathit{BL}_1\cup \{r\}$. $\mathcal{R}$ then moves to $r$ from
$u$ by the memorized port.

The worst-case number of edge traversals occurs in Case
$A\rightarrow D$. By Lemma~\ref{LSP}, this number is not greater
than
$\mathit{LS}(d_1)+O(\Delta^{d_1})\bigl(\mathit{LS}(d_1-1)+\mathit{LS}(d_2-d_1+2)\bigr
)=O(\Delta^{d_2+2})$. For the memory of $\mathcal{R}$, in the
worst case ($A\rightarrow D$), the robot records a path of length
$d_2+2$ and maintains a constant number of variables, therefore
the space is $O(d_2\log\Delta)$ bits.

\end{proof}

\begin{lemma}
Let $u\notin \mathit{BL}_1\cup\{r\} $ be a black node, and let $P$
be a white path from $u$. Let the robot know to which class $u$
belongs. $\mathit{Check\_Par\_Path}(u,P)$ returns ``true" if $P$
is the parent path of $u$ and returns ``false" otherwise. For
$u\in \mathit{BL}_1\cup \{r\}$, let $\mathcal{R}$ take $u$ as a
class $D$ node. $\mathit{Check\_Par\_Path}(u,P)$ returns ``true"
for any path $P$ from $u$ to $r$ containing one edge. When the
procedure exits, $\mathcal{R}$ is in node $u$. There are at most
$O(\Delta^{d_2+2})$ edge traversals in a call to
$\mathit{Check\_Par\_Path}$, and the memory space of $\mathcal{R}$
is $O(d_2\log\Delta)$ bits. \label{ChkSon}

\end{lemma}
\begin{proof}
Procedure $\mathit{Check\_Par\_Path}$ is similar to
$\mathit{Get\_Par\_Path}$ except that the PE part of
$\mathit{Check\_Par\_Path}$ is performed in decreasing
lexicographic order. By Lemma~\ref{GetPar}, for $u\notin
\mathit{BL}_1\cup\{r\} $, providing the robot knows in which class
$u$ is, $\mathit{Check\_Par\_Path}(u,P)$ will find a path in
$\mathit{pred\_path}(u)$ that is lexicographically smaller than
$P$ if this path exists. Whenever a path in
$\mathit{pred\_path(u)}$ is found by
$\mathit{Check\_Par\_Path}(u,P)$, $P$ is not the parent path of
$u$ according to the definition of parent path. $\mathcal{R}$
returns to $u$ via the recorded path. If no such path is found,
$P$ is the minimal path in $\mathit{pred\_path(u)}$, i.e., the
parent path of $u$. $\mathcal{R}$ returns to $u$ by
Lemma~\ref{LSP}. Therefore $\mathit{Check\_Par\_Path}(u,P)$ can
tell whether $P$ is the parent path of $u$. For any $x\in
\mathit{BL}_1\cup \{r\}$, $\mathit{Is\_B}(x)$ returns ``$D$", and
thus $\mathit{C\_or\_D}(x)$ returns ``$D$". Therefore,
$\mathit{Check\_Par\_Path}(u,P)$ returns ``true" for any path $P$
from $u$ to $r$ of length 1. The time and space complexity is
similar as $\mathit{Get\_Par\_Path}$.
\end{proof}

\begin{lemma}
Let $u\neq r$ be a black node, and let $P$ be a white path from
$u$. Let $P'$ be the minimal child path of $u$ greater than $P$ if
this path exists, and let $\mathcal{R}$ know to which class $u$
belongs. Procedure $\mathit{Next\_Child\_Path}(u,P)$ returns the
end of $P'$ if $P'$ exists, and $\mathcal{R}$ is in the end node
of $P'$ when the procedure exits. If $P'$ does not exist, then
$\mathit{Next\_Child\_Path}$ returns ``false" and $\mathcal{R}$
moves to $u$. There are at most $O(\Delta^{2d_2+2})$ edge
traversals in $\mathit{Next\_Child\_Path}$, and the memory space
of $\mathcal{R}$ is $O(d_2\log\Delta)$ bits.\label{GetChd}
\end{lemma}
\begin{proof} Let $\mathcal{R}$ start from node $u\in \mathit{BL}_i$ ($i\geq 1$). We first discuss
four cases in the $\mathit{Enumerating}$ procedure. In each case,
two parts are to be proved: (1) PE can enumerate all paths in
$succ\_path(u)$ that are greater than $P$; (2) NC can identify
whether the end nodes of the paths reported by PE are in
$succ(u)$.

\textbf{ Case(1)}. $u\in C$ ($C\rightarrow D$)

In this case, $succ(u)$ is a subset of the neighbors of $u$, since
$d(\mathit{BL}_{i},\mathit{BL}_{i+1})=1$. Therefore, all paths in
$succ\_path(u)$ that are greater than $P$ can be enumerated by PE.

The neighbors of $u$ belong to class $B$ ($\mathit{BL}_{i-1}$),
class $C$ ($\mathit{BL}_{i}$), or class $D$ ($\mathit{BL}_{i+1}$).
Let $x$ be in $succ(u)$. By Lemma~\ref{isb}, if and only if $x$ is
not a \textbf{B}-node, $\mathit{Is\_B}(x)$ returns ``$D$". By
Lemma~\ref{cord}, if and only if $x$ is a \textbf{B}-node,
$\mathit{C\_or\_D}(x)$ returns ``$D$". Thus NC can identify
whether $x$ is in $succ(u)$. \if 0
 If $\mathit{Is\_B}(x)$ returns
``$D$", by Lemma~\ref{isb}, $x$ is in $succ(u)$. If
$\mathit{Is\_B}(x)$ returns ``\textbf{B}-node", then $x$ is a
\textbf{B}-node. If $\mathit{C\_or\_D}(x)$ returns ``$D$", by
Lemma~\ref{cord}, $x$ is in $succ(u)$. Let $y$ is in $succ(u)$. By
Lemma~\ref{isb}, if $y$ is not a \textbf{B}-node,
$\mathit{Is\_B}(y)$ returns ``$D$". By Lemma~\ref{cord}, if $y$ is
a \textbf{B}-node, $\mathit{C\_or\_D}(y)$ returns ``$D$". Thus NC
can identify whether $x$ is in $succ(u)$.\fi

 \textbf{ Case(2)}. $u\in D$ ($D\rightarrow A$)

In this case, $d(\mathit{BL}_i,\mathit{BL}_{i+1})=d_2$, the local
search from $u$ within radius $d_2$ can report all paths in
$succ\_path(u)$ that are greater than $P$ and their end nodes.

Any black node reported by PE belongs to either class $A$ or class
$D$. Only the nodes in $A$ belong to $succ(u)$. From the
observations of $\mathcal{AL}$ ( i.e., any node in class $D$ has
at least one \textbf{B}-node neighbor, but any node in class $A$
has none), $D\rightarrow A$ identifies whether a reported node is
in $succ(u)$.

 \textbf{ Case(3)}. $u\in A$ ($A\rightarrow B$)

For $d(\mathit{BL}_i,\mathit{BL}_{i+1})=d_1$, the local search
from $u$ within radius $d_1$ can report all paths in
$succ\_path(u)$ that are greater than $P$ and their end nodes.

All the nodes in $succ(u)$ can be reported by $\mathit{LS}_1$. For
the reported black nodes, only the nodes in class $B$ are in
$succ(u)$. According to Lemma~\ref{aorb}, $u$ is in class $B$ iff
$\mathit{A\_or\_B}(x)$ returns ``$B$". Thus by calling
$\mathit{A\_or\_B}$ the nodes in $succ(u)$ can be identified.

 \textbf{ Case(4)}. $u\in B$ ($B\rightarrow C$)

For $succ(u)$ are the neighbors of $u$, $LS_1$ can report all
paths in $succ\_path(u)$ greater than $P$ and their end nodes. Any
layer $C$ node is a black node without white neighbors. Thus NC
can identify whether the end nodes of the paths reported by PE are
in $succ(u)$.

In all above cases, if a node reported by PE is identified as a
node in $succ(u)$ by NC, then the path $P'$ reported by PE is in
$succ\_path(u)$.

Now we consider the $\mathit{Identifying}$ procedure. Let $x$ be
the node returned by $\mathit{Enumerating}$. According to
Lemma~\ref{ChkSon}, $\mathit{Check\_Par\_Path}(x,P'^{-1})$ can
tell whether $P'^{-1}$ is the parent path of $x$, and if so,
$\mathcal{R}$ returns to $x$. Therefore, the minimal child path of
$u$ that is greater than $P$ will be identified if it exists. If
it does not exist, all paths reported by PE do not pass
$\mathit{Identifying}$. $\mathcal{R}$ returns to $u$ in the end,
and the procedure returns ``false".

Denote by $T_{X\rightarrow Y}$ the number of edge traversals of
each case in $\mathit{Next\_Child\_Path}$ and
$\mathit{Get\_Par\_Path}$. By Lemma~\ref{LSP}, Case $D\rightarrow
A$ has the maximal number of edge traversals that is
$T_{D\rightarrow A}\leq
\mathit{LS}(d_2+2)+O(\Delta^{d_2})T_{A\rightarrow
D}=O(\Delta^{2d_2+2})$. For the memory of $\mathcal{R}$, in the
worst case ($D\rightarrow A$), the robot records two paths of
length $d_2$ and $d_2+2$ and maintains a constant number of
variables, thus the space is $O(d_2\log\Delta)$ bits.
\end{proof}

We consider the cases that $r$ is an input of
$\mathit{Next\_Child\_Path}$.
\begin{lemma}
Let $P$ be a path from $u=r$, containing only one edge $e$ ($e$
maybe a self loop). Let $port(e,r)\neq \mathtt{deg}(r)-1$, and let
$\mathcal{R}$ know that $u$ is a class $C$ node.
$\mathit{Next\_Child\_Path}(u,P)$ identifies the path $P'$,
containing only one edge $e'$ from $r$ such that
$port(e',r)=port(e,r)+1$, as the minimal child path of $r$ that is
greater than $P$.\label{DFSroot}
\end{lemma}
\begin{proof}
We can verify that for any $x\in \mathit{BL}_1\cup \{r\}$,
$\mathit{Is\_B}(x)$ returns ``$D$". Thus the following two
statements hold. (1) In the $\mathit{Enumerating}$ part of
$\mathit{Next\_Child\_Path}(r,P)$, the end node $u$ of $P'$ will
be returned. (2) For any path $P''$ from $u$ to $r$ that contains
one edge, $\mathit{Check\_Par\_Path}(u,P'')$ returns ``true". Thus
the lemma is proved.
\end{proof}

The overall exploration performed by our algorithm is the DFSs of
subtrees rooted at each $\mathit{BL}_1$ node of the implicit
spanning tree along with an exploration of
$\mathit{BL}_1\cup\{r\}$. The robot starts from $r$ and explores
each node in $\mathit{BL}_1$ and then explores the subtree rooted
at each node. By
Lemma~\ref{GetPar},~\ref{ChkSon},~\ref{GetChd},~\ref{DFSroot},
starting from any $x\in \mathit{BL}_1$, the robot can conduct a
DFS of the subtree rooted at $x$. Since the robot can identify
nodes in $\mathit{BL}_1$, it can identify whether the DFS of a
subtree is finished. If there are multi-edges between $r$ and $x$,
the subtree of $x$ will be explored more than once. If $r$ has
self loops, $r$ will be identified as a $D$ layer node but without
any child in the spanning tree. In DFSs all the black nodes will
be visited. For any white node $y$, let $\mathit{BL}_i$ be the
black layer such that $d(r,\mathit{BL}_i)<d(r,y)$ and
$\ell=d(y,\mathit{BL}_i)$ is the minimal. By Property~\ref{back},
there exists $u\in \mathit{BL}_i$ such that there is an
$\ell$-white-path from $u$ to $y$. Thus, the PE procedure of
$\mathit{Next\_Child\_Path}(u)$ in a DFS will visit $y$.
Therefore, all white nodes will also be visited by DFSs, and thus
all the nodes in $G$ will be visited. The robot stops once the
exploration is completed, i.e., the robot returns to $r$ via the
largest port at $r$.

\subsection{Bound on the Number of Edge Traversals}
By Lemma~\ref{GetPar},~\ref{GetChd}, the maximal number of edge
traversals of one call to an exploration procedure is
$O(\Delta^{2d_2+2})$. In the DFS, when the robot moves from node
$u$ to $\mathit{parent}(u)$ through the parent path $P$ of $u$ in
state $\mathtt{up}$, the robot has to move back to $u$ to search
for the minimal child path greater than $P^{-1}$. The total number
of edge traversals of these moving backs is not greater than
$bd_2$ where $b=o(n)$ is the number of black nodes. By
Lemma~\ref{DFSroot}, the edges from $r$ are all identified as
child paths in the DFS. If there are $q$ edges between $r$ and a
$\mathit{BL}_1$ node $x$, the subtree rooted at $x$ will be
traversed $q$ times. Denote by $T_{all}$ the total number of edge
traversals by the robot. We have $T_{all}\leq
\Delta(O(\Delta^{2d_2+2})+d_2)o(n)=O(\Delta^{2d_2+3}n)+o(d_2\Delta
n)$.

For simple graphs, the repetitive traversals can be avoided. When
using our algorithm to explore a simple graph labeled by the 1-bit
labeling scheme of~\cite{Ilcinkas08} ($\langle r,2,4\rangle$), the
total number of edge traversals by the robot is $O(\Delta^{10}n)$
which is similar to that in~\cite{Ilcinkas08}.

Given an $\mathcal{AL}$ labeling $\langle r,d_1,d_2 \rangle$ on
$G$ with $L$-ratio $\rho'$. If there are six black layers and
$D=d_1+d_2+3$, the labeling has the minimal $L$-ratio
$\frac{d_1+d_2+4}{6}$, i.e., $\rho'\geq\frac{d_1+d_2+4}{6}$, so
$d_1+d_2\leq 6\rho'-4$. For $d_1\geq 2$, we have $2d_2+3 \leq
2(d_1+d_2)-1\leq 12\rho'-9$. Thus our exploration algorithm
completes in time $O(\Delta^{12\rho'-9}n)+ o(\rho'\Delta n)$.
Since no more than a constant number of paths need to be stored at
the same time and the length of such a path is not greater than
$d_2$, $O(d_2\log \Delta)=O(\rho'\log \Delta)$ bits of memory is
necessary for the robot to explore the graph.

\section{Exploration While Labeling\label{labeling}}

We present an algorithm allowing the robot to label the graph
according to an $\mathcal{AL}$ labeling. As in~\cite{Ilcinkas08},
we assume that before labeling, the graph nodes are labeled by an
initial color named ``blank" that the robot can identify. The
labeling algorithm takes as input an $\mathcal{AL}$ labeling
$L=\langle r,d_1,d_2 \rangle$ and labels the black layers in
order. Denote by $G_i$ the subgraph of graph $G$ induced by all
nodes at distance at most $d(r,\mathit{BL}_i)$ from the root. In
phase $i$ ($i\geq 2$) of the algorithm, the robot starts from the
root and traverses all nodes in $G_i$ and colors the nodes in
$\mathit{BL}_i$ black and colors the nodes in layers between
$\mathit{BL}_{i-1}$ and $\mathit{BL}_{i}$ white. At the end of
phase $i$, the robot has colored $G_i$ according to $L$ and
returned to the root. During the labeling, the labeling algorithm
labels each node only once.

In phase $i$, we call $\mathit{BL}_{i-1}$ the \emph{border layer},
nodes in the border layer \emph{border nodes}, and the set of
nodes that are in the layers between $\mathit{BL}_{i-1}$ and
$\mathit{BL}_{i}$ the \emph{working interval}. In this section, we
always use $\mathit{BL}_{i-1}$ to denote the border layer.

Initially, the robot labels the root (phase 0) and its neighbors
black (phase 1). It then returns to the root. In phase $i$ ($i\geq
2$), if the border layer belongs to class $A$ or $D$, the labeling
procedure has two stages:
%\begin{enumerate}
\vspace{4pt}

\noindent(1) The robot colors all nodes in $\mathit{BL}_i$ black
and returns to $r$.

\noindent(2) The robot colors all nodes in the working interval
white and returns to $r$. \vspace{4pt}
%By Jijun: may need to use \begin{enumerate}...\end{enumerate}
%\end{enumerate}

\noindent If the border layer belongs to class $B$ or $C$, there
is only stage 1. We use $X.x$ to denote the stage of the labeling
algorithm in which the border layer belongs to class $X$ and the
stage is $x$. A 3-bit variable $\mathtt{stage}$ is used to store
the stage, initialized to $D.1$ in phase $2$.

The labeling algorithm includes two procedures: (1) the
exploration procedure; (2) the labeling procedure. The exploration
procedure is a revision of the exploration procedure in
Section~\ref{exp}. In a stage of phase $i$, the robot identifies
some border nodes by the exploration procedure and calls the
labeling procedure from each of these nodes to label the blank
nodes. After calling the labeling procedure from a node, the robot
sets state to $\mathtt{up}$ and moves up to the parent of the
node. When the robot returns to $r$ from the largest port,
variable $\mathtt{stage}$ transforms according to the following
diagram.
\begin{figure*}[ht]%ht
   \hbox to\textwidth{\hfil\includegraphics[width=0.5\textwidth]{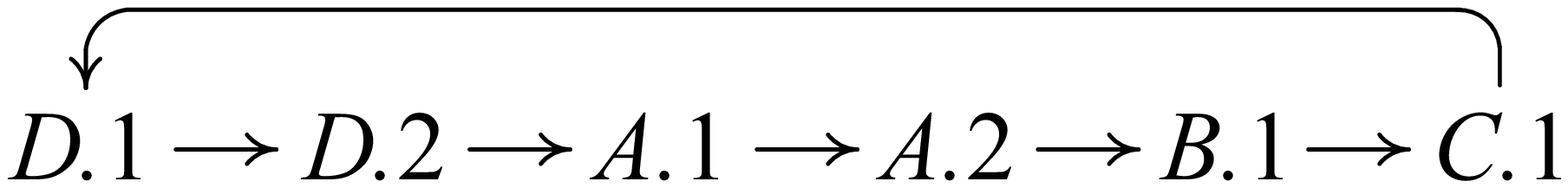}\hfil}%
   \label{A2B}%
\end{figure*}%\

\subsection{Labeling the Nodes}

The robot uses the ${Label\_Succ}$ procedure to color nodes. In
stage $*.1$, for a node $u$ in the border layer, ${Label\_Succ}$
colors all nodes in $succ(u)$ black. In stage $A.2$ and $D.2$,
procedure ${Label\_Succ}$ colors all nodes in the working interval
white.

${Label\_Succ}$ accepts a parameter: $u$, a node in the border
layer. The detail of ${Label\_Succ}$ is given in the following
where we consider six cases.

 (1) \textbf{$\mathtt{stage}=B.1$} or $C.1$. The robot
labels all blank neighbors of $u$ black.

 (2) \textbf{$\mathtt{stage}=D.1$}. The robot performs a
local search from $u$ within radius $d_2$. For each reported blank
node $x$, it performs a local search from $x$ within radius
$d_2-1$. If all black nodes visited in the local search have no
\textbf{B}-node neighbor, then the robot colors $x$ black.

 (3) \textbf{$\mathtt{stage}=D.2$}. The robot performs a
local search from $u$ within radius $d_2$. For every visited blank
node, the robot colors the node white.

 (4) \textbf{$\mathtt{stage}=A.1$}. The robot performs a
local search from $u$ within radius $d_1$. For every blank node
$x$ reported, it performs a local search from $x$ within radius
$d_1-1$. If all black nodes visited in the local search do not
have any white neighbor, then the robot colors $x$ black.

 (5) \textbf{$\mathtt{stage}=A.2$}. The robot performs a
local search from $x$ within radius $d_1$. For every visited blank
node, the robot colors the node white.

\subsection{ Revising the Exploration Procedure }

We revise the exploration procedure in Section~\ref{exp} to
explore the colored subgraph and color the uncolored subgraph. In
a local search, when we say that the robot \emph{ignores} a node,
we mean that as soon as the robot moves in the node, it leaves
this node by the port from which it moves in, not visiting any
neighbor of the node, and continues the local search. The
revisions are given as follows.

The revised $\mathit{Get\_Par\_Path}$ procedure ignores all blank
nodes it visited. The revised $\mathit{Next\_Child\_Path}$
procedure ignores all blank nodes it visited except in the case
where $\mathit{Next\_Child\_Path}(u,P)$ visits a blank node in the
case $X\rightarrow Y$ and $\mathtt{stage}=X.*$. In this case, the
robot returns to $u$ and calls ${Label\_Succ}(u)$.

Table~\ref{table2} gives the operation that the robot performs for
each case of the $\mathit{Next\_Child\_Path}$ procedure when
visiting a blank node in different stages. When ${Label\_Succ}(u)$
terminates, the robot is in node $u$, and it then backtracks to
$\mathit{parent}(u)$ with state $\mathtt{up}$ and continues the
exploration.

\begin{table*}[!]
\caption{ In each case, for each stage the robot performs an
operation
 when visiting a blank node. ``-" means that the robot will not visit a blank node in a combination of a case and a stage.
``$\heartsuit$" denotes the operation to return to $u$ and call
${Label\_Succ}(u)$. ``$\diamondsuit$" denotes the operation to
ignore the blank node.
}%
\begin{center}
\begin{tabular}{|c|c|c|c|c|c|c|}
\hline
Case $\setminus$ $\mathtt{stage}$  & D.1 & D.2 & A.1 & A.2 & B.1 & C.1 \\
\hline $D\rightarrow A$ & $\heartsuit$ & $\heartsuit$ &$\diamondsuit$  &$\diamondsuit$ & - & - \\
\hline $A\rightarrow B$ & - & $\diamondsuit$ &$\heartsuit$  &$\heartsuit$ & $\diamondsuit$ & - \\
\hline $B\rightarrow C$ & $\diamondsuit$ & $\diamondsuit$ &- &$\diamondsuit$ & $\heartsuit$ & $\diamondsuit$ \\
\hline $C\rightarrow D$ & $\diamondsuit$ & $\diamondsuit$ & -  & - & $\diamondsuit$ & $\heartsuit$ \\

\hline
\end{tabular}
\end{center}
\label{table2}%\label{table}%
\end{table*}%

\subsection{Correctness}

For a black node $u$, let $u\in \mathit{BL}_k$,
$rd=d(\mathit{BL}_k,\mathit{BL}_{k+1})$, denote by ${wdisc}(u)$
the set of nodes that the robot visits in a white local search
within radius $rd$ from $u$. It is easy to verify that the whole
graph $G$ is colored according to a labeling scheme $L$, if all
the nodes in ${wdisc}(u)$ are colored according to $L$ for any
black node $u$. A black node $u\in L_k$ $(k\geq 0)$ that has no
neighbor in $L_{k+1}$ is called a \emph{leaf node}.

We prove the correctness of the ${Label\_Succ}$ procedure in the
following.

\begin{lemma} Let $u\in \mathit{BL}_{i-1}$, and let $G_{i-1}$ $(i\geq 2)$ be colored according
to an $\mathcal{AL}$ labeling $L$. If $\mathtt{stage}=B.1$ or
$C.1$ or $A.1$ or $D.1$ and some nodes in $\mathit{BL}_i$ are
colored black, ${Label\_Succ}(u)$ colors all blank nodes in
${succ}(u)$ black, not coloring any other nodes. If
$\mathtt{stage}=A.2$ or $D.2$ and $\mathit{BL}_i$ is colored
according to $L$ and some nodes in the working interval are
colored white, ${Label\_Succ}(u)$ colors all blank nodes in
${wdisc}(u)$ white, not coloring any other nodes.\label{Label_n}
\end{lemma}
\begin{proof}
Let the border layer belong to class $C$ or $B$, and
$\mathtt{stage}=C.1$ or $B.1$ accordingly. Since $G_{i-1}$ and
part of $\mathit{BL}_i$ are colored according to $L$, all blank
neighbors of $u$ are in $succ(u)$. ${Label\_Succ}(u)$ only labels
all blank neighbors of $u$ black. Therefore, ${Label\_Succ}(u)$
colors all blank nodes in ${succ}(u)$ black, not coloring any
other nodes.

Let $u\in D$, and $\mathtt{stage}=D.1$. Let $x$ be a blank nodes
reported by the local search from $u$ within radius $d_2$. In this
case, $G_{i-1}$ and part of $\mathit{BL}_i$ are colored according
to $L$, and all nodes in the working interval are blank. If $x\in
succ(u)$, nodes at distance not greater than $d_2-1$ from $x$ are
either blank nodes or black nodes in $\mathit{BL}_i$; otherwise
there is at least one node in $\mathit{BL}_{i-1}$ at distance less
than $d_2-1$ from $x$ by Property~\ref{back}. Layer
$\mathit{BL}_{i-1}$ is a $D$-layer in which every node has a
\textbf{B}-node neighbor, while each node in $\mathit{BL}_i$ has
no \textbf{B}-node neighbor. Therefore, ${Label\_Succ}(u)$ can
determine whether $x$ is in $succ(u)$. So ${Label\_Succ}(u)$
colors all blank nodes in ${succ}(u)$ black, not coloring any
other nodes.

Let $u\in A$, and $\mathtt{stage}=A.1$. Let $x$ be a blank node
visited by the local search from $u$ within radius $d_1$. If $x\in
succ(u)$, all nodes at distance not greater than $d_1-1$ from $x$
are either blank nodes or black nodes in $\mathit{BL}_i$;
otherwise some of these nodes may belong to $\mathit{BL}_{i-1}$.
Since $G_{i-1}$ has been colored according to $L$ and
$\mathit{BL}_{i-1}$ is an $A$-layer, every node in
$\mathit{BL}_{i-1}$ has a white neighbor. For nodes in the working
interval are blank in stage $A.1$, any node in $\mathit{BL}_i$ has
no white neighbor. By this observation, ${Label\_Succ}(u)$ can
determine whether $x$ is in $succ(u)$. So ${Label\_Succ}(u)$
colors all blank nodes in ${succ}(u)$ black, not coloring any
other nodes.

If $G_{i-1}$ and $\mathit{BL}_i$ are colored according to $L$ and
$\mathtt{stage}=D.2$ or $A.2$, by definition,
$\mathit{Label\_Succ}(u)$ colors all blank nodes in ${wdisc}(u)$,
not coloring any other nodes.

\end{proof}

Now we prove the correctness of the labeling algorithm.
\begin{Theorem}
By the end of the execution of the labeling algorithm taking as
input an $\mathcal{AL}$ labeling $L=\langle r,d_1,d_2 \rangle$,
the graph is fully colored according to $L$, and the robot has
explored the entire graph, terminating at the root.
\label{labelalg}
\end{Theorem}
\begin{proof}
For each $i\geq 0$, we say that $\mathit{Property}(i)$ holds at
the end of phase $i$, if \vspace{4pt}

\noindent (1) The robot colors all nodes of $G_i$ according to $L$
and returns to the root.

\noindent (2) Only nodes of $G_i$ are colored. \vspace{4pt}

We now prove that, at the end of phase $i$, $\mathit{Property}(i)$
holds. Initially, $\mathit{Property}(1)$ holds at the end of phase
1. For $i\geq 1$, assume that at the end of phase $i-1$,
$\mathit{Property}(i-1)$ holds. We prove that
$\mathit{Property}(i)$ holds at the end of phase $i$.

By the induction hypothesis, during phase $i$, all nodes of
$G_{i-1}$ are colored according to the labeling $L$, and all other
nodes are blank.

We first prove that in $X\rightarrow Y$ of
$\mathit{Next\_Child\_Path}$ from $u$, if a blank node is visited
and $\mathtt{stage}=X.*$ then $u$ is a border node. By definition,
for $v\in \mathit{BL}_s$, $\mathit{Next\_Child\_Path}$ from $v$
will not visit any node in $L_t$ such that
$t>d(r,\mathit{BL}_{s+2})$. For a class $X$ node $v\in
\mathit{BL}_s$ $(s<i-1)$, we have $s\leq (i-1)-4$, since all blank
nodes are in layers after $\mathit{BL}_{i-1}$,
$\mathit{Next\_Child\_Path}$ from $v$ will not visit any blank
node. Therefore, according to Table~\ref{table2}, $u$ is a border
node if ${Label\_Succ}(u)$ is called. The robot returns to
$\mathit{parent}(u)$ with state $\mathtt{up}$ when
${Label\_Succ}(u)$ terminates.

Suppose that for $u\in \mathit{BL}_{i-1}$, in a call to
$\mathit{Next\_Child\_Path}$ from $u$, the robot does not visit
any blank node and arrives at a child of $u$ say $v$. By
definition, all neighbors of $v$ will be visited in
$\mathit{Next\_Child\_Path}$ from $u$. Therefore, node $v$ has no
neighbor in layer after $\mathit{BL}_{i}$, and $v$ is a leaf node.
In the followed exploration from $v$, blank nodes will be ignored
(see Table~\ref{table2}).
$\mathit{Next\_Child\_Path}(v,\varnothing)$ returns ``false", and
$\mathcal{R}$ will return from $v$ to $u$ in state $\mathtt{up}$.
If all calls to $\mathit{Next\_Child\_Path}$ from $u$ do not find
a blank node then all children of $u$ are leaf nodes, which
implies that ${wdisc}(u)$ has been colored according to $L$, and
$\mathcal{R}$ returns to $\mathit{parent}(u)$ in state
$\mathtt{up}$ not visiting any node beyond $G_i$.

By the above argument, for $u\in \mathit{BL}_{i-1}$, no mater
whether ${Label\_Succ}(u)$ is called, the robot will return to
$\mathit{parent}(u)$ in state $\mathtt{up}$. For $u\notin
\mathit{BL}_{i-1}$, the blank nodes will be ignored in
explorations from $u$ (see Table~\ref{table2}). For $G_{i-1}$ is
colored correctly, by Theorem~\ref{main}, in phase $i$, all border
nodes are visited, and finally the robot returns to the root.

In the end of phase $i$, for $u\in \mathit{BL}_{i-1}$, either
${Label\_Succ}(u)$ is called or ${wdisc}(u)$ has been colored
according to $L$. By Lemma~\ref{Label_n}, for every border node
$u$, nodes in ${wdisc}(u)$ are colored correctly. Therefore all
the nodes in $G_i$ is colored correctly. Since $\bigcup_{u\in
\mathit{BL}_{i-1}}\mathit{wdisc}(u)\subseteq G_i$, only nodes of
$G_i$ are colored.

In summary, for each $i\geq 0$, $\mathit{Property}(i)$ holds at
the end of phase $i$. It follows that after $\lceil
(D+1)/(d_1+d_2+2) \rceil$ phases, the robot has fully colored and
explored the entire graph. In the end, the last phase is
performed, in which the robot finds that the exploration and the
coloring are completed.

\end{proof}

\section{Labeling Schemes Enabling Adjusting the Ratio of Black Nodes \label{S4}}
Based on $\mathcal{AL}$ labeling schemes, we introduce the
labeling schemes that allow the adjustment of the $N$-ratio. We
will prove the following in the remaining of Section~\ref{S4}.

\begin{Theorem} There exists a robot with
the property that for any $n$-node graph $G$ of degree bounded by
integer $\Delta$, it is possible to color the nodes of $G$ with
two colors (black and white), while the $N$-ratio is not less than
a given rational number $\rho\in (2,(D+1)/4]$. Using the labeling,
the robot can explore the graph $G$, starting from a node $r$ and
terminating at $r$ after identifying that the entire graph has
been traversed. The robot has $O(\rho\log\Delta)$ bits of memory,
and the total number of edge traversals by the robot is
$O(n\Delta^{\frac{16\rho+7}{3}}/\rho+\Delta^{\frac{40\rho+10}{3}})$.
\label{main2}
\end{Theorem}

In the remainder of the paper, word ``\emph{ratio}" refers to
``\emph{$N$-ratio}" if not mentioned.

\subsection{From $L$-ratio Tunable to $N$-ratio Adjustable}

We generalize the $\mathcal{AL}$ labeling to the \emph{periodic
layer oriented labeling} (PL in short). A PL labeling of a graph
is composed of a root node and the sets of layers are colored
black and white. A PL labeling colors the graph in a periodic
manner, that is, $L_i$ and $L_{i+p}$ are colored with the same
color where $p$ is the period. We can represent a PL labeling by a
triple $\langle r, p,BL\rangle$, where $r$ is the root, $0<p\leq
D+1$ is an integer denoting the period, and $BL$ is an integer set
on $[0,p-1]$ denoting the black layers within a period. The set of
black layers of the labeling $\langle r, p,BL\rangle$ is
$\{L_i\mid (i \ \mathtt{mod}\ p)\in BL, 0\leq i\leq D\}$. We call
the interval $[ip,(i+1)p-1]$, $i\geq 0$, the $i$th \textit{unit}
of the labeling. For example, the labeling in Figure~\ref{scheme}
can be denoted by $\langle r, 11,\{0,1,7,10\}\rangle$. Let
$S_1=\langle r, p, BL\rangle,S_2=\langle r, p, BL'\rangle$ be two
PL labeling schemes with the same root and period. The
\textit{union} of $S_1$ and $S_2$ is denoted by $S_1\cup
S_2=\langle r, p, BL\cup BL'\rangle$. Denote by $N(L_i)$ the
number of nodes in layer $L_i$ of a labeling scheme $P$. Denote by
$BN(P)$ the number of black nodes in $P$. Denote by
$\rho(P)=n/BN(P)$ the $N$-ratio of the labeling scheme $P$.

We relax some restrictions of the $\mathcal{AL}$ labeling and
define the following:

\noindent\emph{\textbf{Labeling $\mathcal{MP}$.}}
$\mathcal{MP}=\langle r,p,BL\rangle$ is a PL labeling, where
$BL=\{P_B,P_C,P_D,P_A\}$ satisfies the following properties.
\begin{itemize}
\item  $(P_C-P_B)\ \ \mathtt{ mod }\  p=1$, \item $(P_D-P_C)\ \
\mathtt{ mod }\  p=1$, \item $(P_B-P_A)\ \ \mathtt{ mod }\
p=d_{AB}$, $d_{AB}\geq 2$, \item $(P_A-P_D)\ \ \mathtt{ mod }\
p=d_{DA}$,  \item $\lfloor d_{DA}/2\rfloor\geq d_{AB}$.
\end{itemize}

We call $\mathcal{MP}$ labelings the \textit{elementary
labelings}, and any $\mathcal{AL}$ labeling is an elementary
labeling. In an elementary labeling, $p=d_{AB}+d_{DA}+2$. For
convenience, we use quadruple $\langle r,P_A,d_{AB},d_{DA}\rangle$
to denote an elementary labeling, e.g., the $\mathcal{AL}$
labeling can be denoted by $\langle r,
d_{DA}+1,d_{AB},d_{DA}\rangle$. The labeling in
Figure~\ref{scheme} can be denoted by $\langle r, 7,3,6\rangle$.
In this section, all labelings are elementary labelings or
combinations of elementary labelings.

As for $\mathcal{AL}$ labelings, we partition the black nodes in
an elementary labeling to the following four sets:
\begin{center}
\begin{flushleft}
$C=\{v\in V\mid d(r,v)\ \mathtt{ mod }\ p=P_C\}$,\\
$D=\{v\in V\mid d(r,v)\ \mathtt{ mod }\ p=P_D\}$, \\
$A=\{v\in V\mid d(r,v)\ \mathtt{ mod }\ p=P_A\}$,\\
$B=\{v\in V\mid d(r,v)\ \mathtt{ mod }\ p=P_B\}$.
\end{flushleft}
\end{center}

An $\mathcal{AL}$ labeling scheme cannot guarantee that its
$N$-ratio is not less than its $L$-ratio. The following lemma
implies a method to close the gap.

\begin{lemma}
Given a rational number $1\leq \rho\leq D+1$, let $\rho=m/t$,
where $m>0$ and $t>0$ be integers. Let $PS$ be a set of labeling
schemes of $G$ that have the same root, and $|PS|=m$. If
$\sum_{P\in PS}BN(P)=tn$, then there exists $P\in PS$ such that
$\rho(P)\geq \rho$.\label{Ratio_main}
\end{lemma}
\begin{proof} From
$\sum_{P\in PS}BN(P)=tn$, we have

\begin{equation}
 \frac{\sum_{P\in PS}BN(P)}{n}=\sum_{P\in
PS}\frac{BN(P)}{n}=\sum_{P\in PS}\frac{1}{\rho(P)}=t.
\end{equation}

By the pigeonhole principle, for $|PS|=m$, there exists $P\in PS$
such that $1/\rho(P)\leq t/m$. Therefore, there exists $P\in PS$
such that $\rho(P)\geq m/t=\rho$.
\end{proof}

If we find a set of labelings that satisfies
Lemma~\ref{Ratio_main}, then we can find a labeling where the
$N$-ratio is not less than a given rational number. To generate
such labelings, we introduce the circular shifts of a labeling.
For a labeling $P=\langle r, p,BL\rangle$ and an integer $0\leq
l\leq D$, denote $P^l=\langle r, p,BL^l\rangle$, where
$BL^l=\bigl\{i\mid \bigl((i-l) \ \mathtt{mod}\ p\bigr)\in BL,
0\leq i< p \bigr\}$, called a \textit{circular shift} of $P$. We
give some $N$-ratio adjustable labeling schemes as follows.

Let $\rho\in[4, (D+1)/4]$ be an integer, and let $S=\langle
r,4\rho, BL\rangle$, where $\mathit{BL}=\{0,5,6,7\}$. We have that
$\bigcup_{i=0}^{\rho-1}\mathit{BL}^{4i}= [0,4\rho-1]$, and for any
$0\leq i,j\leq\rho-1$, $\mathit{BL}^{4i}\cap
\mathit{BL}^{4j}=\varnothing$. Let
$H=\{S^0,S^4,\ldots,S^{4(\rho-1)}\}$. We have
$\sum_{i=0}^{\rho-1}\mathit{BN}(S^{4i})=n$. By
Lemma~\ref{Ratio_main}, there exists $S^{4j}\in H$ such that
$\rho(S^{4j})\geq \rho$. This method does not work for $\rho=3$.
We give a solution in Figure~\ref{rho} where we use six different
labeling schemes $T_1,\ldots,T_6$ with period 12, and each layer
is colored black by exactly two labeling schemes. Thus
$\sum_{i=1}^{6}\mathit{BN}(T_i)=2n$. By Lemma~\ref{Ratio_main},
there is $T_i\in\{T_1,\ldots,T_6\}$ such that $\rho(T_i)\geq 3$.

\begin{figure}[htbp]
\centering
\begin{minipage}[t]{0.5\linewidth}
   \begin{center}
   \includegraphics[width=\textwidth]{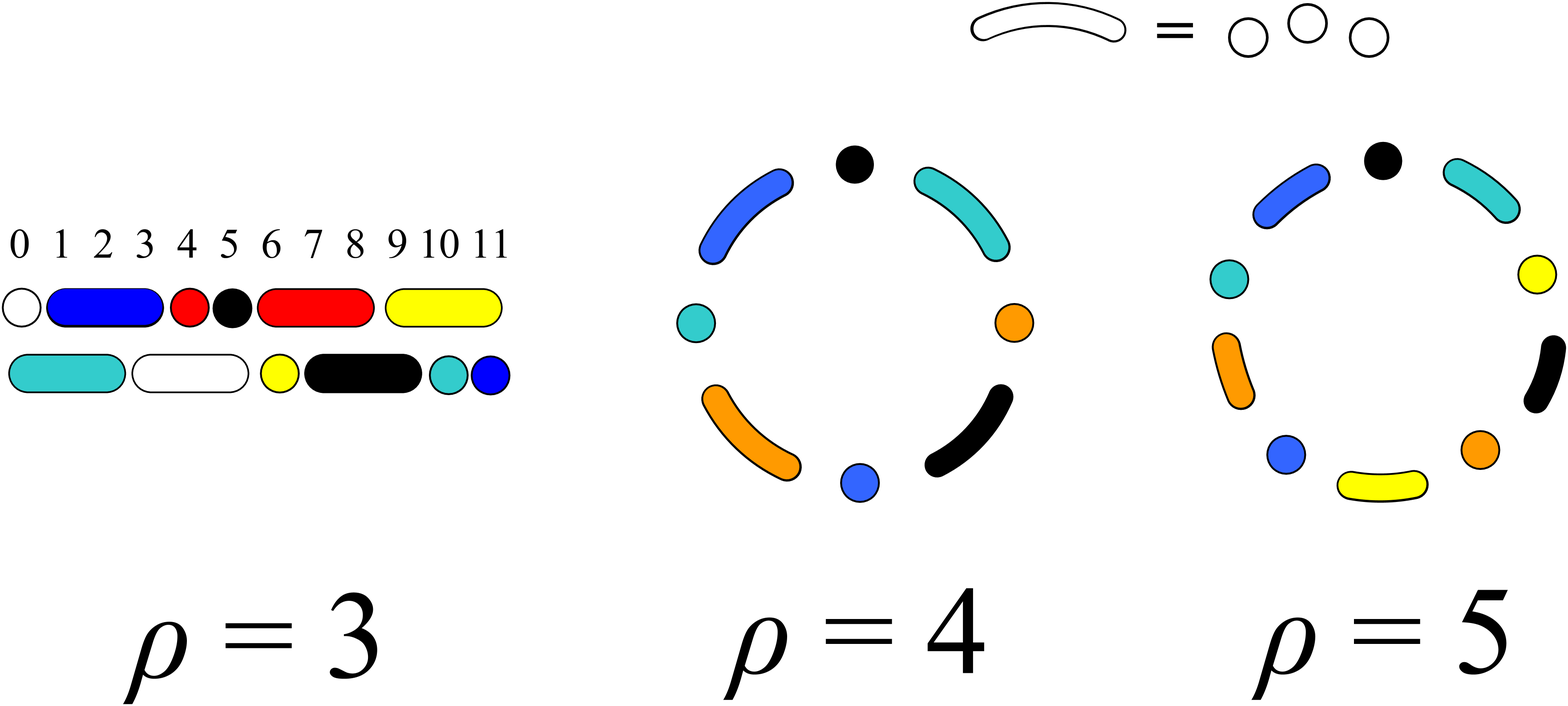}
   \end{center}
   \caption{Examples of adjustable labeling schemes. One unit of each labeling scheme is drawn. A dot represents a layer, and a line represents three adjacent layers.
   Layers colored similarly belong to the same labeling scheme.}%
   \label{rho}%
\end{minipage}
\hspace{10mm}
\begin{minipage}[t]{0.40\linewidth}
   \begin{center}
   \includegraphics[width=\textwidth]{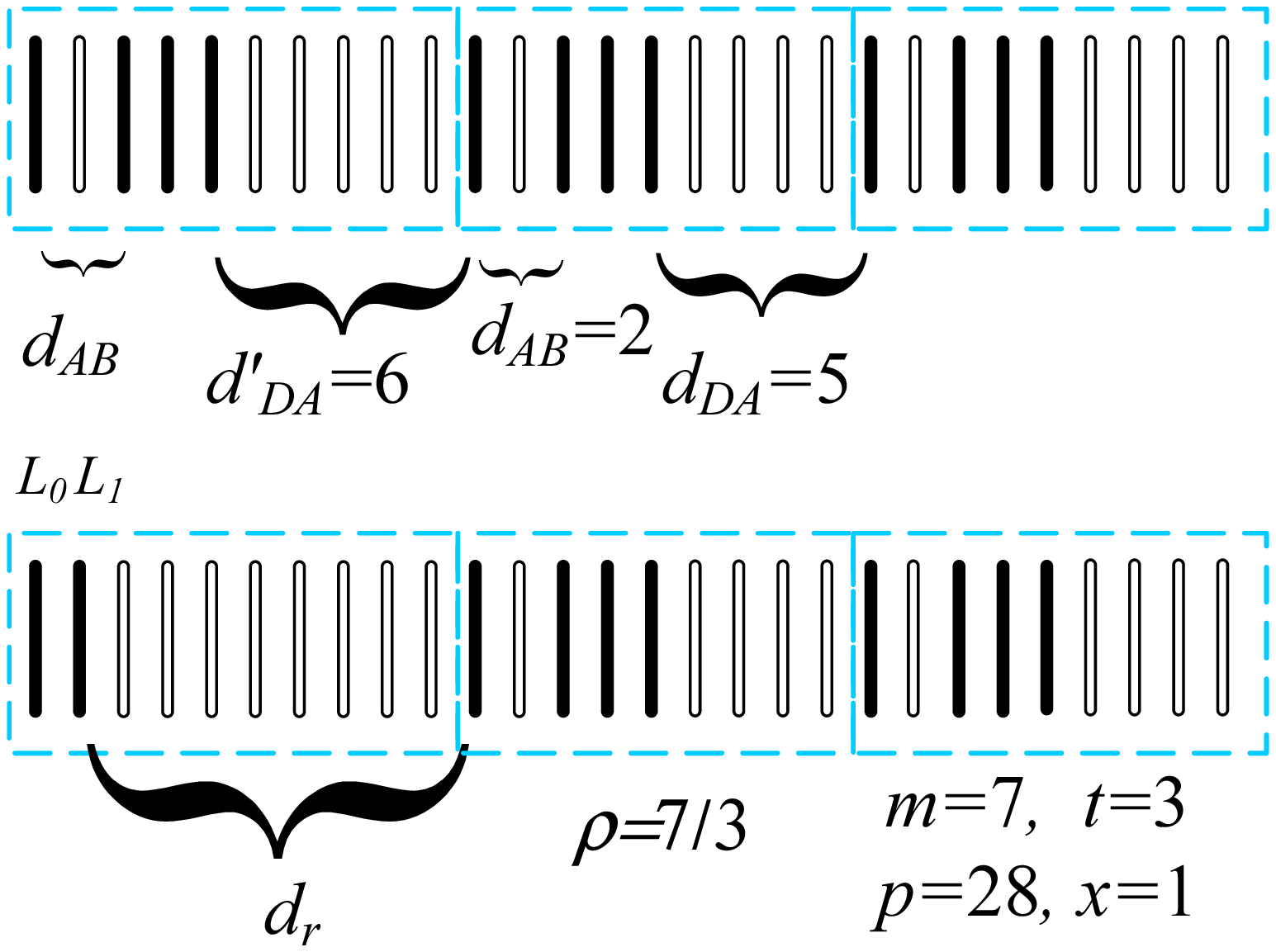}
   \end{center}
   \caption{ Above is a unit of a labeling scheme with a rational $\rho$ where $m=7$, $t=3$. Below is the root unit of $P$. Intervals in a unit are in dotted boxes.}%
   \label{r_rho}%
\end{minipage}

\end{figure}

\subsection{$N$-ratio Adjustable Labeling Schemes}

In this subsection, we introduce a general method to construct the
$N$-ratio adjustable labelings based on Lemma~\ref{Ratio_main}. We
only discuss the cases where $\rho\geq 2$. For $1<\rho< 2$, we can
compute $\rho''=\rho/(\rho-1)\geq 2$. A labeling with $N$-ratio
$\rho$ can be derived from a labeling with $N$-ratio $\rho''$ by
reversing the color of each node.

Given a rational number $\rho\geq 2$, let $\rho=m/t$ where $m$ and
$t$ are relatively prime. The idea is to find a labeling scheme
$P=\langle r, 4m, BL\rangle$ where $|BL|=4t$, $\rho(P)=\rho$. Let
$D+1\geq 4m$, we demonstrate that such $P$ exists. If $m$ is so
huge that $D+1<4m$, we have to find $m'$ and $t'$ that are
relatively primes such that $\rho<m'/t'$ and $D+1\geq 4m'$. Then
we try to find a labeling scheme $P=\langle r, 4m', BL\rangle$
where $|BL|=4t'$, $\rho(P)=m'/t'>\rho$. Let the length of the unit
be $p=4m$, $x=4m\ \mathtt{mod} \ t$. Partition each unit into $t$
disjoint \textit{interval}s. The first $x$ intervals are of length
$\lceil p/t\rceil$, the others are of length $\lfloor p/t\rfloor$.
Let $d_{AB}=\lfloor(\lfloor p/t\rfloor-2)/3\rfloor$,
$d_{DA}=\lfloor p/t\rfloor-2-d_{AB}$, and $d_{DA}'=\lceil
p/t\rceil-2-d_{AB}$. We have $\lfloor d_{DA}/2\rfloor\geq d_{AB}$
and $\lfloor d'_{DA}/2\rfloor\geq d_{AB}$. The following $t$
elementary labelings with the same root and period can be derived.

\begin{displaymath}
S_i = \left\{ \begin{array}{ll}
\langle r,\lceil 4m/t\rceil i,d_{AB},4m-d_{AB}-2 \rangle, &\quad\quad\textrm{$0\leq i\leq x-1$},\\
\langle r,\lceil 4m/t\rceil x+\lfloor
4m/t\rfloor(i-x),d_{AB},4m-d_{AB}-2 \rangle,
&\quad\quad\textrm{$x\leq i\leq t-1$.}
\end{array} \right.
\end{displaymath}

Let $P=\bigcup^{t-1}_{i=0}S_i=\langle r, 4m, BL\rangle$. We have
$|BL|=4t$. In $P$, we classify the nodes into four classes
$A,B,C$, and $D$. Class $X\in\{A,B,C,D\}$ of $P$ is the union of
class $X$ of all $S_i$ ($0\leq i\leq t-1$). There are totally $4m$
circular shift labelings of $P$. For every layer $L_k$, there are
$4t$ circular shift labelings of $P$ where $L_k$ is labeled black.
Therefore, $\sum_{k=0}^{4m-1}BN(P^k)=4tn$. By
Lemma~\ref{Ratio_main}, there exists a circular shift of $P$, say
$P^*$, such that $\rho(P^*) \geq m/t$. An example of $P$ is shown
in Figure~\ref{r_rho}.

\subsubsection{Transformation of $P^*$}

We call the labelings where $r$ is in class~$C$ the \textit{{$R^C$
labelings}}. All $\mathcal{AL}$ labelings are $R^C$ labelings. Let
$P^*$ be the labeling such that $\rho(P^*)\geq \rho$, we can see
that $P^*$ is not necessarily a $R^C$ labeling. In this section,
we give a method to transform $P^*$ which is not a $R^C$ labeling
to a $R^C$ labeling; the exploration algorithm for $\mathcal{AL}$
labelings can be used after minor revisions. The transformation of
$P^*$ is as follows.

Label $r$ and its neighbors black. Let the first $A$-layer in
$P^*$ be $L_l$. If $\lfloor (l-1)/2\rfloor\geq d_\mathit{AB}$, we
label the layers between $L_1$ and $L_l$ white. Otherwise, we
label the layers between $L_1$ and the second $A$-layer white if
this layer exists. If there is only one $A$-layer, then we label
the layers after $L_1$ white. Denote the resulted labeling by
$\hat{P^*}$. We redefine the units of $\hat{P^*}$ as follows: the
interval $[0,d_r]$ is the $0$th unit called {\textit{root unit}},
where $d_r$ is either the distance between the root and the first
$A$-layer if $A$-layers exist or the diameter of $G$ if no
$A$-layers exist; the interval $[(i-1)p+d_r,ip+d_r-1]$ $(i\geq 1)$
is the $i$th \textit{unit}. We have $d_r\geq d_{DA}$.

It is possible that $\mathit{BN}(\hat{P^*})> \mathit{BN}(P^*)$,
and therefore $\rho(\hat{P^*})< \rho$. To make sure
$\rho(\hat{P^*})\geq \rho$, we modify the transformation as
follows. The root is chosen as a node with the minimal number of
neighbors, say $\Delta'$. Label $L_0$ and $L_1$ black. If there
exists an $A$-layer in $P^*$, say $L_k$, such that there is only
one $C$-layer before $L_k$, we label the layers between $L_1$ and
$L_k$ white. If no such $A$-layer exists, the diameter $D$ of $G$
is so short that we label the layers after $L_1$ white.

We prove that $\rho(\hat{P^*})\geq \rho$ as follows. Suppose that
$L_k$ exists. Let $nb_1$ be the total number of black nodes in
layers before $L_k$ in $P^*$, and let $nb_2=N(L_1)+1=\Delta'+1$ be
that in $\hat{P^*}$ which is the number of neighbors of $r$ plus
1. We have $\mathit{BN}(P^*)-\mathit{BN}(\hat{P^*})= nb_1-nb_2$.
In $P^*$, before layer $L_k$, there are a $C$-layer and a
$D$-layer, thus if the root is in a $C$-layer in $P^*$, then $P^*$
and $\hat{P^*}$ are similar; otherwise there are three adjacent
$B$, $C$, and $D$ layers before $L_k$. Because $\Delta'$ is the
minimal number of neighbors of a node
%By Jijun: changed "Fo" to "Because", not sure
in the graph and all the neighbors of nodes in the middle
$C$-layer are involved in the three adjacent black layers, the
number of black nodes in these three layers is not less than
$\Delta'+1=nb_2$. Therefore, $nb_1-nb_2\geq 0$. Thus
$\mathit{BN}(P^*)-\mathit{BN}(\hat{P^*})\geq 0$. We have
$\rho(\hat{P^*})\geq \rho$.

Suppose that $L_k$ does not exist. Since $\rho\leq (D+1)/4$, there
are at least four black layers in the first unit of $P^*$. In this
case, we have $D\leq p+d_2-1$, and in $P^*$ there is only one
$A$-layer, and there are only one $B$-layer and one $C$-layer
after this $A$-layer. When there are three adjacent $B$, $C$, and
$D$ layers after the $A$-layer, based on the above discussions, we
have $\rho(\hat{P^*})\geq \rho$. When there are no three adjacent
$B$, $C$, and $D$ layers after the $A$-layer, the last two layers
are a $B$-layer followed by a $C$-layer. Since all the neighbors
of the nodes in the last $C$-layer are involved in the last two
black layers, the number of black nodes in the last two black
layers is not less than $\Delta'+1=nb_2$. Therefore,
$nb_1-nb_2\geq 0$, and
$\mathit{BN}(P^*)-\mathit{BN}(\hat{P^*})\geq 0$. As a result, we
have $\rho(\hat{P^*})\geq \rho$.

\subsubsection{Exploration Algorithm }

We revise the graph exploration algorithm in Section~\ref{S2} to
explore the graph labeled by $\hat{P^*}$ as follows. First, the
memory of $\mathcal{R}$ increases to $O(d_r\log \Delta)$ bits.
Second, add a 1-bit flag $fr$. If $\mathcal{R}$ is in the root
unit, $fr=1$, otherwise $fr=0$. Third, in the following cases,
$\mathcal{R}$ first determines the distance between a $D$-layer
and the adjacent $A$-layer (we call this distance ``$d_2$" of the
current interval) is $\mathit{d_{DA}}$ or $\mathit{d_{DA}}+1$ or
$d_r$ as follows.

(1) $D\rightarrow A$ of $\mathit{Next\_Child\_Path}$.

Assume that $\mathcal{R}$ is currently in a $D$-layer node $u$. If
$fr=1$, we set $d_2=d_r$ and execute the procedure. If
$D\rightarrow A$ succeeds we set $fr=0$.

Let $fr=0$. We first determine whether $u$ is a leaf node; if not,
we determine the distance between a $D$-layer and the adjacent
$A$-layer. Then we backtrack from $u$ or call $D\rightarrow A$
with the correct $d_2$. The distinguishing procedure is as
follows. Perform a local search from $u$ within radius $d_{DA}$
($\mathit{\mathit{LS}}_1$). If a black node in $A$ is visited then
$d_2=d_{DA}$. If no class $A$ node is visited then perform a local
search from $u$ within radius $d_{DA}+1$
($\mathit{\mathit{LS}}_2$). If a black node $v$ in $A$ is visited
then perform a local search from $v$ within radius $\lceil
d_{DA}/2\rceil$\footnote[6]{If $d_{DA}$ is even, $\lceil
d_{DA}/2\rceil=d_{DA}/2$. If $d_{DA}$ is odd, $\lceil
d_{DA}/2\rceil=\lfloor d_{DA}/2\rfloor+1$}. For each white node
$x$ reported, check whether $R_W(x)=\lfloor d_{DA}/2\rfloor-1$,
and if so, perform a local search within radius $\lfloor
d_{DA}/2\rfloor$ from $x$. If a node with a \textbf{B}-node
neighbor is reported, we have $d_2=d_{DA}$ and $u$ is a leaf node;
otherwise $d_2=d_{DA}+1$. If no class $A$ node is found in
$\mathit{LS}_1$ and $\mathit{LS}_2$ then $u$ is a leaf node.

(2) $A\rightarrow D$ of $\mathit{Get\_Par\_Path}$.

Assume that $\mathcal{R}$ is currently in an $A$-layer node $u$.
We first set $d_2=d_{DA}$ and call the procedure. If $A\rightarrow
D$ fails to find the parent of $u$, we set $d_2=d_{DA}+1$ and redo
$A\rightarrow D$. If it fails again, we set $d_2=d_r$ and $fr=1$
and redo the procedure.

Now we consider the space and the time complexity of the
exploration algorithm. For $d_{DA}=\lfloor p/t\rfloor-2-d_{AB}$,
we have $d_{DA}=\lfloor 4\rho\rfloor-2-\lfloor(\lfloor
4\rho\rfloor-2)/3\rfloor\leq\frac{8\rho-4}{3}$. If $L_0$ is a
$D$-layer in $P^*$, then $\hat{P^*}$ has the maximal $d_r$. In
this case, $d_r=d_{DA}+1+\lceil p/t\rceil\leq\frac{20\rho-1}{3}$.
Thus, the memory of $\mathcal{R}$ is still
$O(\rho\mathrm{log}\Delta )$. Since $d_r\geq d_{DA}$, the number
of edge traversals in exploring the root unit is increased
comparing with $\mathcal{AL}$ labelings. The increased number of
traversals is
$O(\Delta^2\Delta^{2d_{r}+2})=O(\Delta^{\frac{40\rho+10}{3}})$.
The total number of edge traversals is
$O(n\Delta^{2(d_{DA}+1)+3}/\rho+\Delta^{2d_{r}+3})=O(n\Delta^{\frac{16\rho+7}{3}}/\rho+\Delta^{\frac{40\rho+10}{3}})$.

\subsection{Labeling Algorithm}
We use the algorithm in Section~\ref{labeling} with minor
revisions to label a graph according to $\hat{P^*}$. The
parameters of $\hat{P^*}$ are determined by system designers,
including: $r$, $d_{AB}$, $d_{DA}$, and $d_r$. The robot takes as
input these parameters and labels the graph. The revisions of the
exploration procedures are as follows. When $\mathcal{R}$ explores
from a $D$-layer node or an $A$-layer node, the robot has to know
whether the distance from the $D$-layer to the adjacent $A$-layer
(denoted by $d_2$) is $d_{DA}$ or $d_{DA}+1$. We define a variable
$c$ of $\lg t$ bits to indicate that $\mathcal{R}$ is in the $c$th
interval in a unit. Let there be $j$ intervals before the first
$A$-layer of $\hat{P^*}$ in the first unit of ${P^*}$. According
to the definition of $\hat{P^*}$, $d_2=d_{DA}+1$ if $(c+j+t) \
\mathtt{mod}\ t< 4m \ \mathtt{mod}\ t$, and $d_2=d_{DA}$
otherwise. In this description, all arithmetic operations are
modulo $t$. Initially, variable $c$ is set to $t-1$. $c$ increases
by 1 after $\mathcal{R}$ traversed from a $D$-layer node down to
an $A$-layer node and decreases by 1 after $\mathcal{R}$ traversed
from an $A$-layer node up to a $D$-layer node.

When starting from a class-$A$ node or a class-$D$ node, by $c$
and $fr$, the robot knows exactly $d_2$ of the current interval.
So the original exploration procedures in Subsection~\ref{exp} can
be used to explore the graph when $c$ and $fr$ is introduced.

Procedure ${Label\_Succ}$ does not need revision, since the robot
knows $d_2$ of the current interval. Then using the revised
exploration algorithm in this section, the labeling algorithm in
Section~\ref{labeling} can label the graph according to
$\hat{P^*}$.

\section{Future Work}
Further interesting questions include whether there exist labeling
schemes that are not spanning tree based, and whether there exists
a labeling algorithm for an $\mathcal{AL}$ labeling that only uses
two colors. The parameters of the $N$-ratio adjustable labeling
scheme, i.e., the root, are determined by system designers. A
question is whether there exists a finite state automaton that
takes as input a valid $N$-ratio and labels the graph accordingly.

\section*{Acknowledgement}
The authors thanks Leszek G\c{a}sieniec and the anonymous referees
for their constructive suggestions.

\end{document}